\documentclass[runningheads]{llncs}
\usepackage{graphicx}
\usepackage{cite}
\usepackage{amsmath}
\usepackage{amssymb}
\usepackage{mathtools}
\usepackage[unicode]{hyperref}
\usepackage{subcaption}
\usepackage{stackrel}
\usepackage{enumitem}
\usepackage{bussproofs}
\usepackage{wrapfig}

\usepackage{stackengine} 
\usepackage{trimclip}    
\usepackage{xspace}
\usepackage{multirow}
\usepackage{booktabs}
\usepackage[table,xcdraw, usenames,dvipsnames]{xcolor}
\usepackage{times}
\usepackage{subcaption}

\usepackage{tikz}
\usetikzlibrary{shapes}
\usetikzlibrary{positioning}
\usetikzlibrary{fit}
\usetikzlibrary{arrows, decorations.pathmorphing}
\usetikzlibrary{hobby}
\usetikzlibrary{calc}

\def\orcidID#1{}

\newif\iftechreport

\techreporttrue

\newif\ifpaper
\iftechreport
    \paperfalse
\else
    \papertrue
\fi

\begin{document}

\title{Deciding Boolean Separation Logic via Small Models
    \thanks{The work was supported by the Czech Science Foundation project
    GA23-06506S. Basic research funding of the Czech team was provided by the FIT BUT
    internal project FIT-S-23-8151 and the ERC.CZ project LL1908. Tomáš Dacík
    was supported by the Brno Ph.D. Talent Scholarship funded by the Brno City
    Municipality.}
}
\iftechreport
    \subtitle{(Technical Report) \vspace*{-4mm}}
\fi
\iftechreport
    \titlerunning{Deciding Boolean Separation Logic via Small Models (Technical Report)}
\else
    \titlerunning{Deciding Boolean Separation Logic via Small Models}
\fi
    
%
\author{Tomáš Dacík\inst{1}\orcidID{0000-0003-4083-8943} \and Adam
Rogalewicz\inst{1}\orcidID{0000-0002-7911-0549} \and Tomáš
Vojnar\inst{1}\orcidID{0000-0002-2746-8792} \and Florian
Zuleger\inst{2}\orcidID{0000-0003-1468-8398} }
\authorrunning{T. Dacík, A. Rogalewicz, T. Vojnar, and F. Zuleger}

%
\institute{
    Faculty of Information Technology, Brno University of Technology, Czech Republic \and
    Faculty of Informatics, Vienna University of Technology, Austria}

\newcommand{\todo}[1]{\color{red}{TODO: #1}\color{black}}
\newcommand{\anonymize}[2]{{#2}}

\newcommand{\td}[1]{\textcolor{blue}{#1}}
\newcommand{\fz}[1]{\textcolor{cyan}{#1}}
\newcommand{\tv}[1]{\textcolor{magenta}{#1}}
\newcommand{\ar}[1]{\textcolor{green}{#1}}

\newcommand{\BSL}{BSL\xspace}

\newcommand{\arrow}{\xrightarrow{}}
\newcommand{\dom}{\mathsf{dom}}
\newcommand{\img}{\mathsf{img}}
\renewcommand{\implies}{\xrightarrow{}}
\DeclarePairedDelimiter\ceil{\lceil}{\rceil}
\DeclarePairedDelimiter\floor{\lfloor}{\rfloor}
\newcommand{\wildcard}{\underline{\hspace{0.2cm}}}

\newcommand{\vs}{\overline{v}}
\newcommand{\xs}{\overline{x}}
\newcommand{\ys}{\overline{Y}}

\newcommand{\Vars}{\mathsf{Vars}}
\newcommand{\Sort}{\mathsf{Sort}}
\renewcommand{\S}{S}
\newcommand{\lsSort}{\mathbb{S}}
\newcommand{\dlsSort}{\mathbb{D}}
\newcommand{\nlsSort}{\mathbb{N}}

\newcommand{\sh}{(s,h)}
\newcommand{\Locs}{\mathsf{Loc}}
\newcommand{\locations}{\mathsf{locs}}
\newcommand{\Fields}{\mathsf{Field}}
\newcommand{\f}{\mathsf{f}}
\newcommand{\n}{\mathsf{n}}
\newcommand{\p}{\mathsf{p}}
\renewcommand{\t}{\mathsf{t}}

\renewcommand{\phi}{\varphi}
\newcommand{\true}{\mathsf{true}}
\newcommand{\false}{\mathsf{false}}
\newcommand{\pto}{\mapsto}
\newcommand{\cls}[1]{\langle \mathsf{n}\!: #1 \rangle}
\newcommand{\cdls}[2]{\langle \mathsf{n}\!:#1, \mathsf{p}\!:#2 \rangle}
\newcommand{\cnls}[2]{\langle \mathsf{n}\!:#1, \mathsf{t}\!:#2 \rangle}
\renewcommand{\star}{\ast}
\renewcommand{\bigstar}{\mathop{\scalebox{1.5}{\raisebox{-0.1ex}{$\ast$}}}}
\newcommand{\gneg}{\land \neg}
\newcommand{\nil}{\mathsf{nil}}
\newcommand{\emp}{\mathsf{emp}}
\newcommand{\vars}{\mathsf{vars}}

\newcommand{\ls}{\mathsf{sls}}
\newcommand{\lsp}{\ls_{\geq 1}}
\newcommand{\lspp}{\ls_{\geq 2}}
\newcommand{\lspred}{\ls(x,y)}

\newcommand{\dls}{\mathsf{dls}}
\newcommand{\dlsp}{\mathsf{dls}_{\geq 1}}
\newcommand{\dlspp}{\mathsf{dls}_{\geq 2}}
\newcommand{\dlsppp}{\mathsf{dls}_{\geq 3}}
\newcommand{\dlspred}{\dls(x,y,x',y')}

\newcommand{\nls}{\mathsf{nls}}
\newcommand{\nlsp}{\mathsf{nls}_{\geq 1}}
\newcommand{\nlspp}{\mathsf{nls}_{\geq 2}}
\newcommand{\nlspred}{\nls(x,y,z)}

\newcommand{\pheader}{\overline{v}}

\newcommand{\nesteddom}{\dom_N}
\newcommand{\gpath}[4]{#1: #2 {\;\leadsto}_{#4}\; #3}
\newcommand{\gpathG}[5]{#1: #2 {\;\leadsto}_{#4}\; #3}
\newcommand{\plower}{\ell}
\newcommand{\pupper}{u}

\newcommand{\M}{\mathcal{M}}
\newcommand{\ite}[3]{\mathsf{if} \;(#1)\; \mathsf{then} \;(#2)\; \mathsf{else} \;(#3)}
\newcommand{\iite}[3]{\mathsf{if} \;#1\; \mathsf{then} \;#2\; \mathsf{else} \;#3}
\newcommand{\ifthen}[2]{\mathsf{if} \;(#1)\; \mathsf{then} \;(#2)}
\newcommand{\elseif}[2]{\mathsf{else\;if} \;(#1)\; \mathsf{then} \;(#2)}
\newcommand{\ielse}[1]{\mathsf{else} \;(#1)}

\newcommand{\chunks}{\mathsf{chunks}}
\newcommand{\atoms}{\mathsf{atoms}}
\newcommand{\reduction}{\downarrow}
\newcommand{\reduce}[2]{\reduction^{#1}{\!#2}}
\newcommand{\reducesh}{\reduce{s}{h}}
\newcommand{\bound}{\mathsf{bound}}
\newcommand{\reducemodel}[1]{\reduction^{#1}\!(s,h)}
\newcommand{\A}{\mathsf{\alpha}}

\renewcommand{\L}{\mathsf{L}}
\newcommand{\lslocs}{\mathsf{SLS}}
\newcommand{\dlslocs}{\mathsf{DLS}}
\newcommand{\nlslocs}{\mathsf{NLS}}
\newcommand{\arrnext}{h_\n}
\newcommand{\arrprev}{h_\p}
\newcommand{\arrtop}{h_\t}
\newcommand{\arrfield}{t_\f}
\newcommand{\ES}{D_\lsSort}
\newcommand{\ED}{D_\dlsSort}
\newcommand{\EN}{D_\nlsSort}
\newcommand{\heap}{\mathbf{h}}
\newcommand{\translate}[2]{\mathsf{T}(#1, #2)}
\newcommand{\itranslate}[2]{\mathsf{T}_{#1}^{-1}(#2)}
\newcommand{\ttranslate}[1]{\translate{#1}{D}}
\newcommand{\reach}{\mathsf{reach}}
\renewcommand{\path}{\mathsf{path}}
\newcommand{\allpathBounds}[3]{\mathbb{P}^{\leq n}_{(#1, #2)}\;#3.\;}
\newcommand{\allpathBoundsNoDot}[3]{\mathbb{P}^{\leq n}_{(#1, #2)}\;#3\;}
\newcommand{\allpath}[3]{\mathbb{P}_{(#1, #2)}\;#3.\;}

\newcommand{\footprints}[1]{\mathsf{FP}_{(s,h)}(#1)}
\newcommand{\ffootprints}[1]{\mathsf{FP}_{\M}^{\#}(#1)}
\newcommand{\Footprints}[1]{\mathsf{FP}(#1)}
\newcommand{\FFootprints}[1]{\mathsf{FP}^{\#}(#1)}

\newcommand{\NP}{\mathsf{NP}}
\newcommand{\PSPACE}{\mathsf{PSPACE}}
\newcommand{\NEXP}{\mathsf{NEXP}}
\newcommand{\red}[1]{\mathsf{R}(#1)}

\newcommand{\astral}{\textsc{Astral}\xspace}
\newcommand{\cvc}{\textsc{cvc5}\xspace}
\newcommand{\harrsh}{\textsc{Harrsh}\xspace}
\newcommand{\sloth}{\textsc{Sloth}\xspace}
\newcommand{\sls}{\textsc{SLS}\xspace}
\newcommand{\sss}{\textsc{S2S}\xspace}
\newcommand{\songbird}{\textsc{Songbird}\xspace}
\newcommand{\grasshopper}{\textsc{GRASShopper}\xspace}
\newcommand{\benchexec}{\textsc{Benchexec}\xspace}
\newcommand{\bitwuzla}{\textsc{Bitwuzla}\xspace}
\newcommand{\zzz}{\textsc{Z3}\xspace}
\newcommand{\asterix}{\textsc{Asterix}\xspace}

\renewcommand{\H}{\mathcal{H}}
\newcommand{\corresponds}[1]{\approx_{\phi}}

\newcommand{\bvand}{\mathbin{\&}}
\newcommand{\bvor}{\mathbin{|}}
\newcommand{\shiftLeft}{\ll}
\newcommand{\shiftRight}{\gg}
\newcommand{\bitneg}{\ensuremath{\mathord{\sim}}}


\newcommand{\squplus}{%
  \mathrel{\vbox{\offinterlineskip\ialign{%
    \hfil##\hfil\cr
    $\scriptscriptstyle+$\cr
    \noalign{\kern-1ex}
    $\normalsize\sqcup$\cr
}}}}

\newcommand{\auxneq}{\clipbox{0pt 1.5pt 0pt 1.5pt}{$\neq$}}
\newcommand{\auxnotin}{{$\notin$}}
\newlength{\notinlength}
\settowidth{\notinlength}{$\notin$}

\newlength{\neqlength}
\settowidth{\neqlength}{$\auxneq$}

\newcommand{\musteq}{\mathrel{\makebox{\makebox[\neqlength]{$\bigcirc$}\hspace*{-\neqlength}$=$}}}
\newcommand{\mustneq}{\mathrel{\makebox{\makebox[\neqlength]{$\bigcirc$}\hspace*{-\neqlength}$\auxneq$}}}

\newcommand{\shortpto}{\raisebox{0.7pt}{\scalebox{0.67}{$\hspace{1.pt}\mapsto$}}}
\newcommand{\mustpto}[1]{\mathrel{\makebox{\makebox[\neqlength]{$\bigcirc$}\hspace*{-\neqlength}$\shortpto$}}_{\hspace{1pt}#1}}

\newcommand{\shortleadsto}{\raisebox{0.7pt}{\scalebox{0.67}{$\hspace{1.pt}\leadsto$}}}
\newcommand{\mustpath}[1]{\mathrel{\makebox{\makebox[\neqlength]{$\bigcirc$}\hspace*{-\neqlength}$\shortleadsto$}}_{\hspace{1pt}#1}}

\newcommand{\alloc}{\mathsf{alloc}}
\newcommand{\roots}{\mathsf{roots}}
\newcommand{\paths}{\mathsf{paths}}

\newcommand{\muststar}[1]{\stackMath\mathbin{\stackinset{c}{0ex}{c}{0ex}{\star}{\bigcirc}_#1}}
\newcommand{\mustnotin}{\mathrel{\makebox{\makebox[\neqlength]{$\bigcirc$}\hspace*{-\notinlength}\hspace{-1pt}$\auxnotin$}}}

\newcommand{\chunksize}[1]{\lvert\lvert #1 \rvert \rvert}

\newcommand{\gmin}{G^{\mathrm{\ell}}_\sigma}
\newcommand{\gmax}{G^{\mathrm{u}}_\sigma}

\newcommand{\stackbound}{\mathsf{bound}_s}
\newcommand{\locsbound}{\mathsf{bound}}

\maketitle              

\vspace*{-6mm}
\begin{abstract} We present a novel decision procedure for a
fragment of separation logic (SL) with arbitrary nesting of separating
conjunctions with boolean conjunctions, disjunctions, and guarded negations
together with a support for the most common variants of linked lists. Our method
is based on a model-based translation to SMT for which we introduce several
optimisations---the most important of them is based on bounding the size of
predicate instantiations within models of larger formulae, which leads to a much
more efficient translation of SL formulae to SMT. Through a series of
experiments, we show that, on the frequently used symbolic heap fragment, our
decision procedure is competitive with other existing approaches, and it can
outperform them outside the symbolic heap fragment. Moreover, our decision
procedure can also handle some formulae for which no decision procedure has been
implemented so far. \end{abstract}


\vspace*{-8mm}\section{Introduction}\vspace*{-2mm}

In the last decade, separation logic (SL) \cite{SL,OHearn:BI:01} has become one of the
most popular formalisms for reasoning about programs working with
dynamically-allocated memory, including approaches based on deductive
verification \cite{DBLP:journals/sttt/SummersM20}, abstract interpretation \cite{InvaderCAV08}, symbolic
execution \cite{GilPartI20}, or bi-abductive analysis
\cite{Dino-BiAbd11,SecOrderBiAbd14,Broom}. The key ingredients of SL used in these
approaches include the separating conjunction $*$, which allows modular
reasoning by stating that the program heap can be decomposed into disjoint parts
satisfying operands of the separating conjunction, along with inductive
predicates describing shapes of data structures, such as lists, trees, or their
various combinations. 


The high expressive power of SL comes with the price of high complexity and even
undecidability when several of its features are combined together. The existing
decision procedures are usually limited to the so-called \emph{symbolic heap
fragment} that disallows any boolean structure of spatial assertions.

\enlargethispage{6mm}

In this paper, we present a novel decision procedure for a fragment of SL that
we call \emph{boolean separation logic} (\BSL). The fragment allows arbitrary
nesting of separating conjunctions and boolean connectives of conjunction,
disjunction, and a limited form of negation of the form $\phi \gneg \psi$ called
\textit{guarded negation}. To the best of our knowledge, no existing,
practically applicable decision procedure supports a fragment with such a rich
boolean structure and at least basic inductive predicates. The decision
procedure for SL in \cvc \cite{cvc_sl} supports arbitrary nesting of boolean
connectives (including even unguarded negation, which is considered very
expensive in the context of SL) but no inductive predicates. A support for
conjunctions and disjunctions under separating conjunctions is available in the
backend solver of the \grasshopper verifier \cite{automating_sl,
automating_sl_trees} though not described in the papers. In our experimental
evaluation, we outperform both of these approaches on some benchmarks (and
can decide some formulae beyond the capabilities of both of them).
We further
show that adding guarded negations to BSL makes its satisfiability problem $\PSPACE$-hard.

To motivate the usefulness of the fragment we consider, we now give several
examples when SL formulae with a rich boolean structure are useful.
First, in symbolic execution of heap manipulating programs, one usually needs to
consider functions that involve some non-determinism---typically, at least the
\texttt{malloc} statement has the non-deterministic contract $\{\emp\} \;\texttt{x
= malloc()}\; \{x \pto f \lor (x = \nil \land \emp)\}$ (where $f$ is a fresh variable) stating that when the statement is
started in the empty heap, once it finishes, $x$ is either allocated, or the
allocation had failed and the heap is empty. Such contracts typically need a
dedicated (and usually incomplete) treatment when no support of disjunctions is
available.\footnote{Note that, while the post-condition with a single
disjunction might seem simple, the formulae typically start growing in the
further symbolic execution.}
Further, the guarded negation semantically represents the set of counterexamples
of the entailment $\phi \models \psi$, and hence allows one to reduce entailment
queries to UNSAT checking.
Guarded negation can also be used when one needs to obtain several models of a
formula $\phi$ by joining formulae representing the already obtained models to
$\phi$ using guarded negations.
One can also use the guarded negation to express interesting properties
such as the fact that given a list $\ls(x,y)$ and a pointer $y \pto z$, the
pointer does not point back somewhere into the list closing a lasso. This can be
expressed through the formula $\bigl(\ls(x,y) \gneg \bigl(\ls(x, z) \star \ls(z,
y) \bigr)\bigr) \star y \pto z$.
Finally, boolean connectives can be introduced by translating quantitative separation logic into the classical SL \cite{quantitative_sl}.

In this work, we consider BSL with three fixed, built-in inductive predicates
representing the most-common variants of lists: singly-linked (SLL),
doubly-linked (DLL), and nested singly-linked (NLL). Our results can be easily
extended for their variations such as nested doubly-linked lists of
singly-linked lists and the like, but for the price of manually defining their
semantics in the SMT encoding. We do, however, believe that our approach of
bounding the sizes of models and instantiations of the individual predicates can
be lifted to more complex inductive definitions and can serve as a starting
point for allowing integration of SL with inductive definitions into SMT.

\vspace*{-2mm}\paragraph{Contributions.}

Our approach to deciding BSL formulae is inspired by previous works on
translation of SL to SMT. The early works \cite{automating_sl} and
\cite{automating_sl_trees} translate SL to intermediate theories first.  Our
approach is closer to the more recent approach of \cite{sl_data}, which builds
on small-model properties and axiomatizes reachability through pointer links
directly. We extend the SL fragment considered in \cite{sl_data} by going beyond
the so-called unique footprint property (under which it is much easier to obtain
an efficient translation). Further, we define a more precise way to obtain
global bounds on models of entire formulae, and, most importantly, we modify the
translation of inductive predicates in a way that allows us to encode them
succinctly by computing local bounds on their instantiations. According to our
experiments, this makes the decision procedure efficient and competitive with
the state-of-the-art approaches on the symbolic heap fragment (despite the
increased decisive power).
\ifpaper
    The claims we make in this paper are proven~in~\cite{tech_report}.
\fi

\enlargethispage{6mm}

\vspace*{-2mm}\paragraph{Related work.}

In \cite{decidable_sl}, a proof system for deciding entailments of symbolic
heaps with lists was proposed. This problem was later shown to be solvable in
polynomial time in \cite{sl_graphs} via graph homomorphism checking. A
superposition-based calculus for the fragment was presented in
\cite{SL-superposition}, and a model-based approach enhancing SMT solvers was
proposed in \cite{sl_mod_theories}. In \cite{sl_mod_theories}, a combination of
SL with SMT theories is considered but still limited to the symbolic heap
fragment. A more expressive boolean structure and integration with SMT theories
was developed in \cite{automating_sl} for lists and extended for trees in
\cite{automating_sl_trees} but still without a support for guarded
negations.

Other decision procedures are focusing on more general, \textit{user-defined}
inductive predicates (usually of some restricted form). They are based, e.g., on
\textit{cyclic proof systems} (\textsc{Cyclist}
\cite{Brotherston2012AGC}, \textsc{S2S} \cite{S2S, S2S2}); lemma synthesis (\textsc{Songbird}\cite{10.1145/3158097}); or
automata---tree automata are used in the tools \textsc{Slide}
\cite{SL_tree_automata} and \textsc{Spen} \cite{spen}, and a specialised type of
automata, called \textit{heap automata}, is used in \textsc{Harrsh}
\cite{harrsh}.
These procedures do, however, not support nested use of
boolean connectives and separating conjunctions.

There also exist works on deciding much more expressive fragments of SL such as
\cite{bsr_sl, Guarded-SL, strong-sl, conf/concur/IosifZ23} but they do not lead to practically
implementable decision procedures.

\vspace*{-2mm}\section{Preliminaries}

\vspace*{-1.5mm}\paragraph{Partial functions.}

We write $f : X \rightharpoonup Y$ to denote a \textit{partial function} from
$X$ to~$Y$. For a partial function $f$, $\dom(f)$ and $\img(f)$ denote its
domain and image, respectively; $|f| = |\dom(f)|$ denotes its size, and $f(x) =
\bot$ denotes that $f$ is undefined for $x$. A~restriction $f|_A$ of $f$ to $A
\subseteq X$ is defined as $f(x)$ for $x \in A$ and undefined otherwise. To
represent a finite partial function~$f$, we often use the set notation $f =
\{x_1 \mapsto y_1, \ldots, x_n \mapsto y_n\}$ meaning that $f$ maps each $x_i$
to $y_i$, and is undefined for other values. We call partial functions $f_1$ and
$f_2$ \textit{disjoint} if $\dom(f_1) \cap \dom(f_2) = \emptyset$ and define
their \textit{disjoint union} $f_1 \uplus f_2$ as $f_1 \cup f_2$, which is
otherwise undefined.

\vspace*{-1.5mm}\paragraph{Graphs and paths.}

Let $G = (V, \arrow_1, \ldots, \arrow_m)$ be a directed graph with vertices~$V$
and edges $\arrow = \arrow_1 \cup \cdots \cup \arrow_m$. For $1 \leq \f \leq m$,
a sequence $\sigma = \langle v_0, v_1, \ldots, v_n \rangle \in V^{+}$ is a path
from $v_0$ to $v_n$ via $\arrow_\f$ in $G$, denoted as
$\gpathG{\sigma}{v_0}{v_n}{\f}{G}$, if all elements of $\sigma$ are distinct, and
for all $0 \leq i < n$, it holds that $v_i \arrow_\f v_{i+1}$. By the
definition, paths cannot be cyclic. The \textit{domain} of the path~$\sigma$ is
the set $\dom(\sigma) = \{v_0, v_1, ..., v_{n-1}\}$, and the length of the path
is defined as $|\sigma| = |\dom(\sigma)| = n$.

\vspace*{-1.5mm}\paragraph{Formulae.}

For a first-order formula $\phi$, we denote by $\phi[t /
x]$ the formula obtained by simultaneously replacing all free occurrences of the
variable $x$ in $\phi$ with the term $t$.
For a first-order model~$\M$ and a term $t$, we
write $t^{\M}$ to denote the evaluation of $t$ in $\M$ defined as usual.


\enlargethispage{6mm}

\vspace*{-2mm}\section{Separation Logic}\vspace*{-1mm} \label{section:sl}

\vspace*{-1mm}\paragraph{Syntax.}

Let $\Vars$ be a countably infinite set of \textit{sorted variables}. We denote
by $x^S$ a variable~$x$ of a sort $S \in \Sort = \{\lsSort, \dlsSort,
\nlsSort\}$ representing a location in an SLL, DLL, or NLL,
respectively. We omit the sorts when they are not relevant or clear from the
context. We further assume that there exists a distinguished, unsorted
variable~$\nil$. We write $\vars(\phi)$ to denote the set of all variables in $\phi$ plus $\nil$ (even when it does not appear in $\phi$). Analogically, $\vars_S(\phi)$ stands for all variables of the sort $S$ plus $\nil$.

The syntax of our fragment is given by the following grammar:
\begin{align*}
    p   &\Coloneqq x^{\lsSort} \mapsto \cls{n} \;|\;
    x^{\dlsSort} \mapsto \cdls{n}{p} \;|\;
    x^{\nlsSort} \mapsto \cnls{n}{t} && \text{(points-to predicates)}\\
    \pi &\Coloneqq \ls(x^{\lsSort}, y^{\lsSort}) \;|\;
    \dls(x^{\dlsSort}, y^{\dlsSort}, x^{\dlsSort}_b, y^{\dlsSort}_b) \;|\;
    \nls(x^{\nlsSort}, y^{\nlsSort}, z^{\lsSort}) && \text{(inductive predicates)}\\
    \phi_{A} &\Coloneqq x = y \;|\; x \neq y \;|\; p \;|\; \pi
    && \text{(atomic formulae)}\\
    \phi &\Coloneqq \phi_{A} \;|\; \phi \star \phi \;|\; \phi \land \phi \;|\;
    \phi \lor \phi \;|\; \phi \gneg \phi && \text{(formulae)}\vspace*{-0.5mm}
\end{align*}

The \emph{points-to} predicate $x \mapsto \langle \f_1
: f_1, \ldots, \f_n: f_n\rangle$ denotes that $x$ is a structure whose fields
$\f_i$ point to values $f_i$. We often write $x \mapsto n$ instead of $x \mapsto
\cls{n}$ and $x \mapsto \wildcard$ if the right-hand side is not relevant.
We call $x$ the \emph{root} of the points-to predicate.
If $\pi$ is an inductive predicate $\lspred$, $\dlspred$, or $\nlspred$, we
again call $x$ the root of $\pi$, $y$ is the \emph{sink} of $\pi$, and we write
$\pi(x, y)$ to denote the root and the sink. We define the \emph{sort} of the
predicate $\pi$, denoted as $S_\pi$, as the sort of its root. Then, there is a
one-to-one correspondence of predicates and sorts, which we often implicitly
use.


\vspace*{-1.5mm}\paragraph{Memory model.}

Let $\Locs$ be a countably infinite set of memory locations, and let $\Fields =
\{ \n, \p, \t \}$ be the set of fields. A \textit{stack} is a finite partial
function $s: \Vars \rightharpoonup \Locs$. A \textit{heap} is a finite partial
function $h : \Locs \rightharpoonup (\Fields \rightharpoonup \Locs)$. For
succinctness, we write $h(\ell, \f)$ instead of $h(\ell)(\f)$. To represent heap
elements in a readable way, we write functions $\Fields \rightharpoonup \Locs$
as vectors with labels, i.e., $h(\ell) = \langle \f : h(\ell, \f) \;|\;  \f \in \Fields \;\land\; h(\ell, \f) \neq
\bot \rangle$ and we write $\img(h)$ for $\{\ell \in \Locs \;|\; \exists \ell', \f.\; h(\ell', \f) = \ell\}$. Moreover, we use $h(\ell) = n$ when
$h(\ell) = \langle \n : n \rangle$. A \textit{stack-heap model} is a pair
$(s,h)$ where $s$ is stack and $h$ is a heap such that $s(\nil) \neq \bot$ and $h(s(\nil)) = \bot$.
We define the set of locations of the model $(s,h)$ as $\locations{(s,h)} = \img(s) \cup \dom(h) \cup \img(h)$.

\begin{figure}[t]
\vspace*{-6mm}
\begin{align*}
    (s,h) &\models x \bowtie y
        &\text{iff }& s(x) \bowtie s(y) \text{ and } \dom(h) = \emptyset
        \text{ for} \bowtie \;\in \{=, \neq\}\\
    (s,h) &\models x \mapsto \langle \f_i : f_i \rangle_{i \in I}
        &\text{iff }& h = \{s(x) \mapsto \langle \f_i : s(f_i) \rangle_{i \in I}\} \\
    (s,h) &\models \psi_1 \bowtie \psi_2
        &\text{iff }& (s,h) \models \psi_1 \bowtie (s,h) \models \psi_2
        \text{ for} \bowtie \;\in \{\land, \gneg, \lor\}\\
    (s,h) &\models \psi_1 \star \psi_2
        &\text{iff }& \exists h_1, h_2.\; h = h_1 \uplus h_2 \neq \bot
        \text{ and } (s,h_i) \models \psi_i \text{ for } i =1, 2\\
\\
    (s,h) &\models \exists x.\; \psi
        &\text{iff }& \text{there exists $\ell$ such that }
        (s[x \mapsto \ell],h) \models \psi\\
    (s,h) &\models \lspred
        &\text{iff }& (s,h) \models x = y \text{, or } s(x) \neq s(y)\\
        &&&\text{and } (s, h) \models \exists n.\; x \mapsto n \star \ls(n, y)\\
    (s,h) &\models \dlspred
        &\text{iff }& (s,h) \models x = y \star x' = y' \text{, or }
        s(x) \neq s(y), s(x') \neq s(y'),\\
        &&& \text{and } (s,h) \models \exists n.\; x \mapsto \cdls{n}{y'}
        \star \dls(n, y, x', x) \\
    (s,h) &\models \nlspred
        &\text{iff }& (s,h) \models x = y \text{, or } s(x) \neq s(y)\\
        &&&\text{and } (s, h) \models \exists n, t.\; x \mapsto \cnls{n}{t}
        \star \ls(n, z) \star \nls(t, y, z)
\end{align*}
\vspace*{-6mm}
  \caption{The semantics of the separation logic. The existential quantifier is
  used for the definition of the semantics of inductive predicates and it is not
  a part of our fragment.}
\vspace*{-3mm}
\label{fig:sl_semantics}
\end{figure}

\enlargethispage{6mm}

\vspace*{-1.5mm}\paragraph{Semantics.}

The semantics of our SL over stack-heap models is given in
Fig.~\ref{fig:sl_semantics}. For pure formulae, we use the so-called
\textit{precise semantics}, which additionally requires that the heap must be
empty\footnote{This is a common approach to avoid the atom $\true$ to be
expressed as $\nil = \nil$. In our fragment, we forbid $\true$ in order not to
introduce ``unbounded'' negations as $\neg \phi \triangleq \true \gneg \phi$. Due to this change, symbolic heaps are formulae of form $\bigstar \psi_i$ where each $\psi_i$ is an atom.}.
The semantics of pointer assertions, boolean connectives, and separating
conjunctions is as usual. The intuition behind the semantics of the inductive
predicates is as follows.
An SLL segment $\ls(x,y)$ is either empty or represents an acyclic sequence of
allocated locations starting from $x$ and leading via the $\n$ field to $y$,
which is not allocated.
A DLL segment $\dls(x, y, x', y')$ is either empty with $x = y$ and $x' = y'$,
or it represents an acyclic sequence that is doubly-linked via the $\n$ and $\p$
fields and leads from the first allocated location $x$ of the segment to its
last allocated location $x'$ ($x$ and $x'$ may coincide) with $y$/$y'$ being the
$\n$/$\p$-successors of $x'$/$x$, respectively. Both $y$ and $y'$ are not allocated.
An NLL segment $\nls(x, y, z)$ is a (possibly empty) acyclic sequence of
locations starting from $x$ and leading to $y$ via the $\t$ (top) field in which
successor of each locations starts a disjoint inner SLL to $z$ via $\n$.

\vspace*{-1mm}\paragraph{Stack-heap graphs.}

We frequently identify stack-heap models with their graph representation. A
stack-heap model $(s,h)$ defines a graph $G[(s,h)] = (V, (\arrow_\f)_{\f \in \Fields})$ where
$V = \locations(s,h)$ and
$u \arrow_\f v$ iff $h(u, \f) = v$. We frequently use the fact that if
there exists a path $\gpath{\sigma}{x}{y}{\f}$ in a stack-heap graph,
then it is uniquely determined because $\f$-edges are given by a partial function.

\vspace*{-1mm}\section{Small-Model Property}\label{section:small-models}

Small-model properties, which state that each satisfiable formula has a model of bound\-ed size, are frequently used for various fragments of SL to prove their
decidability~\cite{CC_results} or to design decision procedures \cite{cvc_sl,
strong-sl,sl_data}. The latter is also the case of our translation-based
decision procedure which will heavily rely on enumeration over all locations,
and, for its efficiency, it is therefore necessary to obtain location
bounds that are as small as possible.

The way we obtain our small-model property is inspired by the approach of
\cite{sl_data} and by insights from the so-called \textit{strong-separation
logic} \cite{strong-sl}. The main idea is to define a satisfiability-preserving
reduction $\reduce{s}{h}$ which takes a heap~$h$ (referenced from a stack~$s$),
decomposes it into basic sub-heaps (which we call \textit{chunks}), and reduces it per the sub-heaps in such a
way that its size can be easily bounded by a linear expression. To define the
reduction, we first need to introduce some auxiliary notions related to
stack-heap models.


We say that a model $(s,h)$ is \emph{positive} if there exists $\phi$ with
$(s,h) \models \phi$.
A positive model $(s,h)$ is \emph{atomic} if it is non-empty, and for all
positive models $(s, h_1)$ and $(s, h_2)$, $h = h_1 \uplus h_2$ implies that
$h_1 = \emptyset$ or $h_2 = \emptyset$.
In other words, atomic models cannot be decomposed into two non-empty positive
models. Several examples of atomic models are shown in Fig.~\ref{fig:reduction}.
Observe that the models of $\dls$ (Figure \ref{fig:reduction-dls}) and $\nls$
(Figure \ref{fig:reduction-nls}) are indeed atomic as any of their
decomposition, in particular the split at the location $u$, does not give two
positive models.

\begin{figure*}[t]
    \centering
    \begin{minipage}{0.5\textwidth}
        \begin{subfigure}{1\textwidth}
            \centering
            \begin{tikzpicture}[
    node distance={9mm}, 
    thick, 
    main/.style = {draw, circle}]
    
    \tikzstyle{node}=[thick,draw=blue!75,fill=blue!20,minimum size=5mm,rounded corners=0.15cm]
    \tikzstyle{rnode}=[thick,draw=red!75,fill=red!20,minimum size=5mm,rounded corners=0.15cm]
    \tikzstyle{empty}=[]
    
    \node[node]  (1)                 {$x$};
    \node[node]  (2) [right of=1]    {$\;$};
    \node[empty] (l) [above of=2, yshift=-10pt] {$\ell$};
    \node[rnode] (3) [right of=2]    {$\;$};
    \node[rnode] (4) [right of=3]    {$\;$};
    \node[node]  (5) [right of=4]    {$y$};
    
    \draw[-stealth'] (1) -- (2);
    
    \draw[-stealth', dotted] (2) -- (3);
    \draw[-stealth', dotted] (3) -- (4);
    \draw[-stealth', dotted] (4) -- (5);

    \draw[-stealth', preaction={draw, blue!30,-, double=blue!30, double distance=3\pgflinewidth,}] (2) to [out=40, in=140] (5);
    
\end{tikzpicture}
            \caption{A singly-linked list $\lspred$.}
            \label{fig:reduction-ls}
        \end{subfigure}
        \begin{subfigure}{1\textwidth}
            \centering
             \begin{tikzpicture}[node distance={9mm}, thick, main/.style = {draw, circle}]
    
    \tikzstyle{node}=[thick,draw=blue!75,fill=blue!20,minimum size=5mm,rounded corners=0.15cm]
    \tikzstyle{rnode}=[thick,draw=red!75,fill=red!20,minimum size=5mm,rounded corners=0.15cm]
    \tikzstyle{empty}=[]
    
    \node[node]  (1)                 {$y'$};
    \node[node]  (2) [right of=1]    {$x$};
    \node[node]  (3) [right of=2]    {$\;$};
    \node[empty] (l) [above of=3, yshift=-10pt] {$\ell$};
    \node[rnode] (4) [right of=3]    {$u$};
    \node[rnode] (5) [right of=4]    {$\;$};
    \node[node]  (6) [right of=5]    {$x'$};
    \node[node]  (7) [right of=6]    {$y$};

    \draw[-stealth'] (2.200) to (1.340);
    \draw[-stealth'] (2.20) to (3.160);
    \draw[-stealth'] (3.200) to (2.340);
    \draw[-stealth', dotted] (3.20) to (4.160);
    \draw[-stealth', dotted] (4.200) to (3.340);
    \draw[-stealth', dotted] (4.20) to (5.160);
    \draw[-stealth', dotted] (5.200) to (4.340);
    \draw[-stealth', dotted] (5.20) to (6.160);
    \draw[-stealth', dotted] (6.200) to (5.340);
    \draw[-stealth'] (6.20) to (7.160);

    \draw[-stealth', preaction={draw, blue!30,-, double=blue!30, double distance=3\pgflinewidth,}] (3) to [out=50, in=130] (6);
    \draw[-stealth', preaction={draw, blue!30,-, double=blue!30, double distance=3\pgflinewidth,}] (6) to [out=230, in=310] (3);
    
\end{tikzpicture}
            \vspace*{-2mm}
            \caption{A doubly-linked list $\dlspred$.}
            \label{fig:reduction-dls}
        \end{subfigure}
    \end{minipage}
    \begin{minipage}{0.45\textwidth}
        \begin{subfigure}{1\textwidth}
            \centering
             \begin{tikzpicture}[node distance={11mm}, thick, main/.style = {draw, circle}]
    
    \tikzstyle{node}=[thick,draw=blue!75,fill=blue!20,minimum size=5mm,rounded corners=0.15cm]
    \tikzstyle{rnode}=[thick,draw=red!75,fill=red!20,minimum size=5mm,rounded corners=0.15cm]
    \tikzstyle{empty}=[]
    
    \node[node]  (1)                 {$x$};
    \node[node]  (2) [right of=1]    {$\;$};
    \node[empty] (l) [above of=2, yshift=-15pt] {$\ell$};
    \node[rnode] (3) [right of=2]    {$\;$};
    \node[rnode] (4) [right of=3]    {$\;$};
    \node[node]  (5) [right of=4]    {$y$};

    \draw[-stealth'] (1) -- (2);
    \draw[-stealth', dotted] (2) -- (3);
    \draw[-stealth', dotted] (3) -- (4);
    \draw[-stealth', dotted] (4) -- (5);

    \node[empty]  (21) [below of=2]   {$\;$};
    \node[node]   (z)  [below of=21]   {$z$};
    \draw[-stealth'] (2) -- (z);
    
    \node[rnode] (11) [below of=1]   {$\;$};
    \node[rnode] (12) [below of=11]   {$u$};
    \draw[-stealth', dotted] (1) -- (11);
    \draw[-stealth', dotted] (11) -- (12);
    \draw[-stealth', dotted] (12) -- (z);

    \draw[-stealth', preaction={draw, blue!30,-, double=blue!30, double distance=3\pgflinewidth,}] (1) -- (z);
    
    \draw[-stealth', dotted] (3) -- (z);

    \node[rnode] (41) [below of=4]   {$\;$};
    \draw[-stealth', dotted] (4) -- (41);
    \draw[-stealth', dotted] (41) -- (z);

    \draw[-stealth', preaction={draw, blue!30,-, double=blue!30, double distance=3\pgflinewidth,}] (2) to [out=40, in=140] (5);

\end{tikzpicture}
            \vspace{3pt}
            \caption{A nested singly-linked list $\nlspred$.}
            \label{fig:reduction-nls}
        \end{subfigure}
    \end{minipage}

    \caption{An illustration of reductions of atomic models of inductive
    predicates. Removed heap locations are red, removed edges are
    dotted, and added edges are highlighted.}
    \vspace*{-5mm}

  \label{fig:reduction}
\end{figure*}
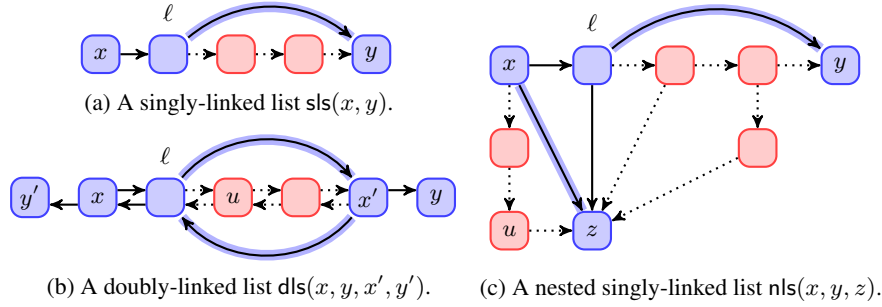


\enlargethispage{6mm}

A sub-heap $c \subseteq h$ is a \emph{chunk} of a model $(s,h)$ if $c$ is a
maximal sub-heap of $h$ such that $(s, c)$ is an atomic positive model.
Notice that the way the definition of chunks is constructed excludes the
possibility of using as a chunk a sub-heap of a heap that itself forms an atomic
model. The reason is that otherwise the remaining part of the larger atomic
model could not described by the available predicates. For example, in nested
lists as shown in Fig.~\ref{fig:reduction-nls}, one cannot take as a chunk a
part of some inner list (e.g., the pointer $u \mapsto z$) as the heap shown in
the figure itself forms an atomic model. Indeed, if  $u \mapsto z$ was removed,
one would need a more general version of the NLL predicate to cover the remaining heap
by atomic models.\vspace*{-1mm}


\begin{lemma}[Chunk decomposition] \label{lemma:chunk-decomposition} A positive
model $(s,h)$ can be uniquely decomposed into the set of its chunks, denoted
$\chunks(s,h)$, i.e., $h=\biguplus\chunks(s,h)$. \end{lemma}

\vspace*{-4mm}\paragraph{Minimal atomic models of inductive predicates.}

The key reason why the small-model property that we are going to state holds is
that our fragment of SL cannot distinguish atomic models of the considered
predicates beyond certain small sizes---namely, two for $\ls$ and $\nls$, and
three for $\dls$.
For further use, we will now state predicates describing exactly the sets of the
indistinguishable lists of the different kinds.

We start with SLLs and use a disequality to exclude empty lists: $\lsp(x,y)  \triangleq \lspred \star x \neq y$, and a guarded
negation to exclude lists of length one consisting of a single pointer only:
$\lspp(x,y) \triangleq
\lsp(x,y) \gneg (x \pto y)$.
A similar predicate can be defined for NLLs too: $\nlspp(x,y,z) \triangleq
\bigl(\nls(x,y,z) \star x \neq y\bigr) \gneg (x \pto \cnls{z}{y})$.


\enlargethispage{6mm}

For DLLs, we 
define $\dlspp(x, y, x',
y') \triangleq \dlspred \star x \neq  y \star x \neq x'$ to exclude models that are either empty or consist of a single pointer; and $\dlsppp(x, y, x', y') \triangleq \dlspp(x, y, x', y') \gneg (x \pto \cdls{x'}{y'} \star x' \pto \cdls{y}{x})$ to also exclude models consisting
of exactly two pointers.

%


It holds that atomic models, and consequently also chunks, are
precisely either models of single pointers or of the above predicates.

\begin{lemma}
\label{lemma:atom-classification}
For atomic model $(s,h)$, exactly one of the following
conditions holds. \begin{enumerate}

  \item $(s,h) \models x \mapsto \wildcard\;$ for some $x$.
  \hfill{(pointer-atom)}

  \item $(s,h) \models \lspp(x,y)$ for some $x$ and $y$. \hfill{($\ls$-atom)}

  \item $(s,h) \models \dlsppp(x, y, x', y')$ for some $x$, $y$, $x'$, and y'.
  \hfill{($\dls$-atom)}

  \item $(s,h) \models \nlspp(x, y, z)$ for some $x$, $y$, and $z$.
  \hfill{($\nls$-atom)}

\end{enumerate}
\end{lemma}

We can now define the reduction in the way we have already sketched.

\begin{definition} The heap of a positive model $(s, h)$ reduces to $\reducesh =
\biguplus_{c \in \mathsf{chunks}(s,h)} \reduce{s}{c}$ where the reduction of a
chunk $c$ with a root $x$ as follows:\begin{itemize}
  \vspace*{-2mm}

  \item $\reduce{s}{c} = c$ if $(s,c) \models x \mapsto \wildcard\;$.

  \item $\reduce{s}{c} = \{s(x) \mapsto \ell, \ell \mapsto s(y) \}$ where $\ell = c(s(x),\n)$
  if \mbox{$(s,c)\!\models\!\lspp(x, y)$ for some $y$}.

  \item $\reduce{s}{c} = \{ s(x) \mapsto \langle \mathsf{n}\!:\!\ell,$
  $\mathsf{p}\!:\!s(y') \rangle, \ell \mapsto \langle \mathsf{n}\!:\!s(x'),
  \mathsf{p}\!:\!s(x) \rangle, s(x') \mapsto \langle \mathsf{n}\!:\!s(y), \mathsf{p}\!:\!
  \ell \rangle \}$ where $\ell = c(s(x),\n)$ if $(s,c) \models \dlsppp(x, y, x', y')$ for some $x'$, $y'$ and
  $y$.

  \item $\reduce{s}{c} = \{s(x) \mapsto \langle \mathsf{t}: \ell, \mathsf{n}: s(z)
  \rangle, \ell \mapsto \langle \mathsf{t}: s(y), \mathsf{n}: s(z)\rangle \}$
  where $\ell = c(s(x),\t)$ if $(s,c) \models \nlspp(x, y, z)$ for some $y$ and $z$.

  \vspace*{-2mm}
\end{itemize}
\end{definition}

We lift the reduction to stack-heap models as $\reducemodel{X} = (s', \reduce{s'}{h})$ where $s' = s|_X$ for some set of variables $X$ and show that it preserves satisfiability when $X = \vars(\phi)$.

\begin{theorem}
\label{theorem:reduction}
For a positive model $(s,h)$, it holds that $(s,h) \models \phi$
iff $\;\reducemodel{\vars(\phi)} \models \phi$. \end{theorem}

\enlargethispage{6mm}

The final step to show our small-model property is to find an upper bound on the size of the reduced models. We define the size of a variable $x^S$, $\chunksize{x^S}$, which represents its contribution to the location bound, and is defined as $2$ if $S \in \{\lsSort, \nlsSort\}$ and $1.5$ if $S = \dlsSort$ (this corresponds to the size of a reduced chunk of sort $S$ divided by the number of variables which are allocated in it). We further define $\chunksize{\nil} = 0$. The location bound of~$\phi$ is then given as $\locsbound(\phi) = 1 + \floor{\sum_{x \in \vars(\phi)} \chunksize{x}}$ (the additional location is for $\nil$). Analogically, the location bound for a sort $S$ is $\locsbound_S(\phi) = \floor{\sum_{x \in \vars_S(\phi)} \chunksize{x}}$.

\begin{theorem}[Small-model property] \label{theorem:small-models}
If a formula $\phi$ is satisfiable, then there exists a model $(s, h) \models \phi$ such that $|\locations(s,h)| \leq \locsbound(\phi)$.
\end{theorem}

We conjecture that the bound can be further improved, e.g., by showing that each
model can be transformed to an equivalent one (indistinguishable by BSL
formulae) such that the number of its chunks is bounded by the number of roots
of spatial predicates in $\phi$. We demonstrate this on the formula $\ls(x, y)
\star y \mapsto z$ and its model in which~$y$ points back into the middle of the
list segment (thus splitting it into two chunks). Clearly, this model can be
transformed by redirecting $z$ \mbox{outside of the list domain.}





\vspace*{-2mm}\section{Translation-Based Decision Procedure}\vspace*{-1mm}
\label{section:translation}

In this section, we present our translation of SL to SMT. We first present an SMT encoding of our memory model and a translation of basic predicates and boolean connectives. Then we discuss methods for efficient translation of separating conjunctions and inductive predicates with the focus on avoiding quantifiers by replacing them by small enumerations of their instantiations.

We fix an input formula $\phi$ and let $n_S = \locsbound_S(\phi)$ for each sort $S \in \Sort$.

\vspace*{-2mm} \subsection{Encoding the Memory Model in SMT} \vspace*{-1mm}
\label{section:encoding-memory-model}

To encode the heap, we use a classical approach which encodes its mapping and domain separately \cite{automating_sl, sl_data, cvc_sl}. Namely, we use arrays to encode mappings and sets to encode domains.
We also use the theory of datatypes to represent a finite sort of locations by a datatype
%
$\L \triangleq\;
    \mathsf{loc}^{\nil}
     \,\,|\,\, \mathsf{loc}^{\lsSort}_1  \;|\; \ldots \;|\; \mathsf{loc}^{\lsSort}_{n_\lsSort}
     \,\,|\,\, \mathsf{loc}^{\dlsSort}_1 \;|\; \ldots \;|\; \mathsf{loc}^{\dlsSort}_{n_\dlsSort}
     \,\,|\,\, \mathsf{loc}^{\nlsSort}_1 \;|\; \ldots \;|\; \mathsf{loc}^{\nlsSort}_{n_\nlsSort}.
$

Now, we define the signature of the translation's language over the sort $\L$. For each $x~\in~\vars(\phi)$, we introduce a constant $x$ of the same name---its interpretation represents the stack image~$s(x)$. To represent the heap, we introduce a set symbol $D$ representing the domain and an array symbol $h_\f$ for each field $\f \in \Fields$ which represents the mapping of the partial function $\lambda \ell.\; h(\ell, \f)$. To distinguish sorts of locations, we further introduce a set symbol $D_S$ for each sort $S \in \Sort$. We define meaning of these symbols by showing how a stack-heap model can be reconstructed from a first-order model.

%
%
%

\begin{definition}[Inverse translation]
\label{def:model_translation}
Let $\M$ be a first-order model. We define its inverse translation $\itranslate{\phi}{\M} = (s,h)$ where $s(x)$ = $x^\M$ if $x \in \vars(\phi)$ and
\begin{align*}
     h(\ell) &= 
     \begin{cases}
         \cls{\arrnext[\ell]^\mathcal{M}} &\text{ if $\ell \in (D \cap \ES)^{\M}$}\\
         \cdls{\arrnext[\ell]^\mathcal{M}}{\arrprev[\ell]^\mathcal{M}} &\text{ if $\ell \in (D \cap \ED)^{\M}$}\\
         \cnls{\arrnext[\ell]^\mathcal{M}}{\arrtop[\ell]^\M} &\text{ if $\ell \in (D \cap \EN)^{\M}$}.
     \end{cases}
 \end{align*}
\end{definition}

To ensure consistency of the translation with the memory model used, we define
the following axioms that a result of translation needs to satisfy:
%
%
%
\begin{align*}
     \mathcal{A}_\phi \triangleq
         \nil = \mathsf{loc}^{\nil} \land\;
         \nil \not \in D \;\land\;
         \!\!\!\!\bigwedge_{S \in \Sort}\!\! 
            \bigl(D_S = \{\mathsf{loc}^{\nil}, \mathsf{loc}^{S}_1, \ldots, \mathsf{loc}^{S}_{n_S}\}
         \;\land\; \!\!\!\!\!\!\!\!\!\bigwedge_{x \in \vars_S(\phi)}\!\!\!\!\!\! x \in D_S\bigl).
\end{align*}
\vspace*{-3mm}

\noindent The axioms ensure that $\nil$ is never allocated, that each variable is
interpreted as a location of the corresponding sort and they fix the
interpretation of the sets $\ES, \ED, \EN$, which we will later use in the
translation to assign sorts to locations.


\vspace*{-3mm}
\subsection{Translation of SL to SMT}
\enlargethispage{6mm}
We define the translation as a function $\mathsf{T}(\phi) = \mathcal{A}_\phi \land \translate{\phi}{D}$ where $\mathcal{A}_\phi$ are the above defined axioms and $\translate{\phi}{D}$ is a recursive translation function of the formula $\phi$ with the domain symbol $D$. The translation $\mathsf{T}(\cdot)$ together with the inverse translation of models $\mathsf{T}^{-1}_\phi(\cdot)$ are linked by the following correctness theorem.

\begin{theorem}[Translation correctness]
\label{theorem:correctness}
An SL formula $\phi$ is satisfiable iff its translation $\mathsf{T}(\phi)$ is satisfiable. Moreover, if $\mathcal{M} \models \mathsf{T}(\phi)$, then $\itranslate{\phi}{\M} \models \phi$.
\end{theorem}

The translation of non-inductive predicates and boolean connectives is defined as:
\begin{align*}
\translate{x \bowtie y}{F} &\triangleq x \bowtie y \;\land\; F = \emptyset \hspace{50pt} \text{for $\bowtie\; \in \{=, \neq\}$}\\
\translate{\psi_1 \bowtie \psi_2}{F} &\triangleq \translate{\psi_1}{F} \bowtie \translate{\psi_2}{F}
\hspace{28pt}\text{for $\bowtie\; \in \{\land, \lor, \gneg\}$}\\
\translate{x \pto \langle \f_i: f_i\rangle_{i \in I}}{F} &\triangleq F = \{x\} \;\land\; \bigwedge_{i \in I} h_{\f_i}[x] = f_i
\end{align*}

\noindent The translation of boolean connectives follows the boolean structure and propagates the domain symbol $F$ to the operands. The translation of pointer assertions postulates content of memory cells represented by arrays and also requires the domain $F$ to be~$\{x\}$. 

\vspace*{-1mm}\paragraph{Translation of separating conjunctions.}

The semantics of separating conjunctions involves a quantification over sets (heap domains). The most direct way of  translation is to use quantifiers over sets leading to decidable formulae due to the bounded location domain. This approach combined with a counterexample-guided quantifier instantiation is used in the decision procedure for a fragment of SL supported in \textsc{cvc5} \cite{cvc_sl}. In some fragments, however, separating conjunctions can be translated in a way that completely avoids quantifiers. An example is the fragment of boolean combinations of symbolic heaps which has the so-called \textit{unique footprint property} (UFP)~\cite{automating_sl, sl_data}---a formula $\psi$ has a (unique) footprint in a model $(s,h)$ with $(s,h) \models \psi \star \true\footnote{Assuming the standard semantics of $\true$ which is not part of our logic.}$, if there exists a (unique) set $F$ such that $(s, h|_{F}) \models \psi$.  The UFP-based approaches of \cite{automating_sl, sl_data} axiomatize the footprints during translation and check operands of separating conjunctions just on the sub-heaps induced by their footprints.

However, UFP does not hold for \BSL because of disjunctions. As an example, take the formula $\psi \triangleq x \mapsto y \lor \emp$ and the heap~$h = \{x \mapsto y\}$. Both $(s,h|_{\{s(x)\}}) \models \psi$ and $(s,h|_{\emptyset}) \models \psi$ hold. The sets $\{s(x)\}$ and $\emptyset$ are, however, the only footprints of~$\psi$ in $(s,h)$, and this observation can be used to generalise the idea of footprints beyond the fragment in which they are unique.

Instead of axiomatizing the footprints, our translation builds a set of footprint terms for operands of separating conjunctions. This change can be also seen as a simplification of the former translations as it eliminates the need to deal with two kinds of formulae (the actual translation and footprint axioms), which must be treated differently during the translation. However, the precise computation of the set of all footprints of $\psi$ in $(s,h)$, denoted as $\mathsf{FP}_{(s,h)}(\psi)$, is as hard as satisfiability---when the set of footprints is non-empty, the formula $\psi$ is satisfiable. Therefore, we compute just an over-approximation denoted as $\FFootprints{\psi}$. This is justified by the following lemma which gives an equivalent semantics of the separating conjunction in terms of footprints.

\begin{lemma}
\label{lemma:star_simplified}
Let $\phi \triangleq \psi_1 \star \psi_2$ and let $(s,h)$ be a model. Let $\mathcal{F}_1$ and $\mathcal{F}_2$ be sets of locations such that $\mathsf{FP}_{(s,h)}(\psi_i) \subseteq \mathcal{F}_i$. Then $(s,h) \models \psi_1 \star \psi_2$ iff
\begin{align*}
\bigvee_{F_1 \in \mathcal{F}_1} \;
\bigvee_{F_2 \in \mathcal{F}_2} \;
\bigwedge_{i = 1, 2} (s, h|_{F_i}) \models \psi_i \;\land\; F_1 \cap F_2 = \emptyset \;\land\; F_1 \cup F_2 = \dom(h).
\end{align*}
\end{lemma}

\noindent Intuitively, to check whether a separating conjunction holds in a model, it is not necessary to check all possible splits of the heap, but only the splits induced by (possibly over-approximated) footprints of its operands. The lemma is therefore a generalisation of UFP and leads to the following definition of the translation $\translate{\psi_1 \star \psi_2}{F}$:
\begin{align*}
\exists F_1 \in \mathcal{F}_1.\;
\exists F_2 \in \mathcal{F}_2.\;\;\;
\translate{\psi_1}{F_1} \land \translate{\psi_2}{F_2} \land F_1 \cap F_2 = \emptyset \land F = F_1 \cup F_2.
\end{align*}

\noindent Here, we use a quantifier expression of the form $\exists x \in X.\; \psi$ as a placeholder that helps us to define two methods which the translation can use for separating conjunctions:

\begin{itemize}
    \item The method $\mathtt{SatEnum}$ computes sets of footprints $\mathcal{F}_i$ as $\FFootprints{\psi_i}$ (the computation is described below) and replaces expressions $\exists x \in X.\; \psi$ with $\bigvee_{x' \in X} \psi[x'/x]$ as in Lemma \ref{lemma:star_simplified}. This strategy is quite efficient in many practical cases when we can compute small sets of footprints $\mathcal{F}_1$ and $\mathcal{F}_2$.

    \item The method $\mathtt{SatQuantif}$ does not compute sets $\mathcal{F}_i$ at all and replaces \mbox{$\exists x \in X.\; \psi$} simply with $\exists x.\; \psi$. This strategy is better when the existential quantifier can be later eliminated by Skolemization  or when the set of footprints would be too large. 
    %
\end{itemize}

We now show how to compute the set of footprint terms $\FFootprints{\psi}$. We again postpone inductive predicates to Section \ref{section:predicate-translation}. We just note that their footprints are unique. The cases of pure formulae and pointer assertions follow directly from the definition of their semantics, which requires the heap to be empty and a single pointer, respectively.
\begin{align*}
    \FFootprints{x \bowtie y}  = \{\emptyset\} \;\text{ for } \bowtie\; \in \{=, \neq \}
    && \FFootprints{x \pto \wildcard} = \{\{x\}\}
\end{align*}

\noindent For the boolean conjunction, we can select from footprints of its operand the one with the lesser cardinality. Since negations have many footprints (consider, e.g., $\neg \emp$), we define the case of the guarded negation by taking footprints of its guard. The disjunction is the only case which brings non-uniqueness as we need to consider footprints of both of its operands.
\begin{align*}
\FFootprints{\psi_1 \gneg \psi_2} &= \FFootprints{\psi_1}
\hspace{40pt} \FFootprints{\psi_1 \lor \psi_2} = \FFootprints{\psi_1} \cup \FFootprints{\psi_2}\\
\FFootprints{\psi_1 \land \psi_2} &= \text{ if } |\FFootprints{\psi_1}| \leq |\FFootprints{\psi_2}| \text{ then } \FFootprints{\psi_1} \text{ else } \FFootprints{\psi_2}
\end{align*}

\noindent Finally, we define footprints of the separating conjunction by taking the union $F_1 \cup F_2$ for each pair $(F_1, F_2)$ of footprints of its operands. Notice that here $F_1 \cup F_2$ represents an SMT term, therefore we cannot replace it with a disjoint union which is not available in the classical set theories in SMT. We can, however, use heuristics and filter out terms for which we can statically determine that interpretations of $F_1$ and $F_2$ are not disjoint.
\begin{align*}
\FFootprints{\psi_1 \star \psi_2} = \{F_1 \cup F_2 \;|\; F_1 \in \FFootprints{\psi_1} \text{ and } F_2 \in \FFootprints{\psi_2}\}
\end{align*}

\enlargethispage{6mm}

We state the correctness of the footprint computation in the following lemma.

\begin{lemma}
\label{lemma:footprints}
Let $\M$ be a first-order model with $\M \models \mathsf{T}(\phi)$ and let $(s,h) = \itranslate{\phi}{\M}$. Then we have
$\mathsf{FP}_{(s,h)}(\phi) \subseteq \{F^\M \;|\; F \in \mathsf{FP}^{\#}(\phi)\}$.
\end{lemma}

\vspace*{-2mm} \subsection{Translation of Inductive Predicates} \vspace*{-1mm}
\label{section:predicate-translation}

To translate inductive predicates, we express them in terms of reachability and
paths in the heaps. While unbounded reachability cannot be expressed in
first-order logic, we can efficiently express bounded \textit{linear}
reachability in our encoding. The linearity means that each path uses only a
single field (which is not the case, e.g., for paths in trees). All predicates
in this section are parametrised with an interval $[m,n]$ which bounds the
length of the considered paths. When we do not state the bounds explicitly, we assume conservative bounds $[0,
\locsbound_S(\phi)]$ for a path starting from a root of a sort $S$. We show how to
compute more precise bounds in Section~\ref{section:optimisations}. We start
with the translation of reachability:\vspace*{-1mm}
\begin{align*}
    \reach^{=n}(h, x, y) \triangleq h^n[x] = y &&
    \reach^{[m,n]}(h, x,y)  \triangleq \bigvee \!_{m \leq i \leq n}
    \reach^{=i}(h, x, y)\vspace*{-1mm}
\end{align*}

\noindent Here, the predicate $\reach^{=n}(h, x, y)$ expresses that $x$ can
reach $y$ via a field represented by the array $h$ in exactly $n$ steps.
Similarly, $\reach^{[m,n]}$ expresses reachability in $m$ to $n$ steps. Besides
reachability, we will need a macro $\path_C(h, x, y)$ expressing the domain of a
path from $x$ to $y$, or the empty set if such a path does not exists:
\begin{align*}
    \path^{=n}_C(h, x,y) \triangleq\;&
        \bigcup\;_{0\leq i < n} \; C(h^i[x])\\
\path^{[m,n]}_C(h, x,y) \triangleq\;&
\ifthen{\reach^{=m}(h, x, y)}{\path^{=m}_C(h, x, y)}\\
    & \cdots \; \elseif{\reach^{=n}(h, x, y)}{\path^{=n}_C(h, x, y)}
    \;\ielse{\emptyset}
\end{align*}
The additional parameter $C$ is a function applied to each element of the path
that can be used to define nested paths. We define a simple path $\path^{[m,n]}_S(h, x,
y) \triangleq \path^{[m,n]}_C(h, x, y)$ with $C \triangleq \lambda \ell.\; \{\ell\}$ and a nested path as $\path^{[m,n]}_N(h_1, h_2, x, y, z) \triangleq \path^{[m,n]}_C(h_1, x, y)$ with $C \triangleq \lambda \ell.\; \path_S(h_2, \ell, z)$. In the case of the nested
path, the array $h_1$ represents the top-level path from $x$ to $y$, and $h_2$
represents nested paths terminating in the common location $z$. Now we can
define footprints of inductive predicates using $\path$ terms as follows:
%
%
\begin{align*}
\FFootprints{\pi(x,y)} &= \{ \path_S(\arrnext, x, y) \} && \text{for $\pi \in \{\ls, \dls\}$}\\
\FFootprints{\nlspred} &= \{ \path_N(\arrtop, \arrnext, x, y, z) \}
\end{align*}

The common part of the translation $\translate{\pi(x,y)}{F}$ postulates the
existence of a top-level path from $x$ to $y$ and a domain $F$ based on this
path (formalised in the formula $\mathsf{main\_path}$ below); and ensures that
all locations have the correct sort (through the formula $\mathsf{typing}$). For
DLLs, we add an invariant which ensures that its locations are correctly
doubly-linked (the $\mathsf{back\_links}$ formula), and we further need a
special treatment of the cases when the list is empty as well as a special
treatment for its roots and sinks (cf. the formula $\mathsf{boundaries}$). For
NLLs, we add an invariant stating that an inner list starts from each location
in its top-level path (the $\mathsf{inner\_lists}$ formula) and that those
inner paths are disjoint (the $\mathsf{disjoint}$ formula).
%
%

\enlargethispage{6mm}

\begin{itemize}

    \item $\translate{\lspred}{F} \;\triangleq\; \mathsf{main\_path} \land
    \mathsf{typing}$ where \begin{align*}
        \mathsf{main\_path} \triangleq \reach(\arrnext, x, y) \land 
        F = \path_S(\arrnext, x, y) \text{  and  }
        \mathsf{typing} \triangleq F \subseteq \ES.
    \end{align*}
    
    \item $\translate{\dlspred}{F} \;\triangleq\; \mathsf{empty} \lor 
    \mathsf{nonempty}$ where\begin{align*}
        \mathsf{empty}  
            &\triangleq x = y \land x' = y' \land F = \emptyset,\\
        \mathsf{nonempty} 
            &\triangleq x \neq y \land x' \neq y' \land \mathsf{main\_path} \land
            \mathsf{boundaries} \land \mathsf{typing} \land
            \mathsf{back\_links},\\
        \mathsf{main\_path} 
            &\triangleq \reach(\arrnext, x, y) \land F = \path_S(\arrnext, x,
            y),\\
        \mathsf{boundaries} 
            &\triangleq \arrprev[x] = y' \land \arrnext[x'] = y \land x' \in F \land
            y' \not \in F,\\
        \mathsf{typing} 
            &\triangleq F \subseteq \ED,\\
        \mathsf{back\_links} 
            &\triangleq \forall \ell.\; (\ell \in F \land \ell \neq x') \implies 
            \arrprev[\arrnext[\ell]] = \ell.
    \end{align*}
    
    \item $\translate{\nlspred}{F} \;\triangleq\;
        \mathsf{main\_path} \land \mathsf{typing} \land \mathsf{inner\_lists} \land
        \mathsf{disjoint}$ where
        \begin{align*}
            \mathsf{main\_path} &\triangleq \reach(\arrtop, x, y) \land 
            F = \path_N(\arrtop ,\arrnext, x, y, z),\\
            \mathsf{typing} &\triangleq \path_S(\arrtop, x, y) \subseteq \EN
                \land F \setminus \path_S(\arrtop, x, y) \subseteq \ES,\\
            \mathsf{inner\_lists} &\triangleq \forall \ell.\; \ell \in F \cap \EN 
            \implies \reach(\arrnext, h[\ell], z),\\
            \mathsf{disjoint} &\triangleq \forall \ell_1, \ell_2.
                   \bigl(\{\ell_1, \ell_2\} \subseteq F \land \ell_1 \neq \ell_2
                   \land \arrnext[\ell_1] = \arrnext[\ell_2]\bigl) \implies
                   \arrnext[\ell_1] \not \in F\text{\footnotemark}.
        \end{align*}

\end{itemize}
\vspace*{-2mm}

\footnotetext{We could also write $\arrnext[\ell_1] = z$ instead of $\arrnext[\ell_1] \not \in F$, but the latter leads to better performance of SMT solvers.}

\vspace*{-1mm}\paragraph{Path quantifiers.} Invariants of paths are naturally
expressed using universal quantifiers. For quantifiers, however, we cannot
directly take advantage of bounds on path lengths. Therefore, similarly as for
separating conjunctions, we use the idea of replacing quantifiers by small
enumerations of their instances, which is efficient when we can compute small
enough bounds on the paths. For example, if we know that the length of an
$\f$-path with a root $x$ is at most two, it is enough to instantiate its
invariant for $x$, $h_\f[x]$, and $h^2_\f[x]$. This idea is formalised using
expressions $\allpathBounds{h}{x}{\ell} \psi$, which we call \textit{path
quantifiers} and which state that $\psi$ holds for all locations of the path
with the length $n$ starting from~$x$ via the array $h$:

\vspace{-4mm}
$$\allpathBounds{h}{x}{\ell} \psi \;\triangleq\; \bigwedge \!_{0 \leq i \leq n}
\;\; \psi[h^i[x]/\ell].$$

\noindent If we need to quantify over nested paths, we need to use two path quantifiers
(one for the top-level path and one for the nested paths). The quantifiers in
the last conjunct of the NLL translation can be rewritten as
$\allpath{h_t}{x}{\ell'_1} \allpath{h_t}{x}{\ell'_2}
\allpath{h_n}{\ell'_1}{\ell_1} \allpath{h_n}{\ell'_2}{\ell_2}$ In this
expression, $\ell'_1$ and $\ell'_2$ range over locations in the top-level list, and
$\ell_1$ and $\ell_2$ range over locations in the nested paths starting from $\ell'_1$ and
$\ell'_2$, respectively.
\vspace{1mm}

%
%

%


\vspace*{-2mm} \subsection{Complexity} \vspace*{-1mm} \label{section:complexity}
\enlargethispage{6mm}

This section briefly discusses the complexity of the proposed decision procedure
as well as the complexity lower bound for the satisfiability problem in the
considered fragment of SL. We will use $\mathsf{SAT}(\omega_1, \ldots,
\omega_n)$ to denote the satisfiability problem for a sub-fragment constructed
of atomic formulae and the connectives $\omega_i$ and
$\mathsf{SAT}(\overline{\omega_1, \ldots, \omega_n})$ to denote the fragment
where none of the connectives $\omega_i$ appear.

\begin{theorem}
\label{theorem:complexity}
The procedure $\mathtt{SatQuantif}$ produces formula of polynomial size, and, for $\mathsf{SAT}(\overline{\gneg})$, it runs in $\NP$. The procedure $\mathtt{SatEnum}$ runs in $\NP$ for $\mathsf{SAT}(\overline{\lor})$.
\end{theorem}

\begin{proof}[sketch] When not considering the instantiation of quantifiers over
footprints, both $\mathtt{SatQuantif}$ and $\mathtt{SatEnum}$ produce a formula
$\mathsf{T}(\phi)$ of a polynomial size dominated by the translation of
inductive predicates. For the variant of the translation of inductive predicates using universal quantifiers over locations, the size is $\mathcal{O}(n^3)$ for SLLs and DLLs
(dominated by the $\mathcal{O}(n^3)$ size of the $\path_S$ term), and
$\mathcal{O}(n^5)$ for NLLs (dominated by $\path_N$). If the input formula does
not contain guarded negations, then all quantifiers can be eliminated using
Skolemization. The translated formulae are then in a theory decidable in $\NP$
(e.g., when sets are encoded as extended arrays \cite{z3_sets}).

The procedure $\mathtt{SatEnum}$ can produce exponentially large formulae
because of the footprint enumeration. This can be prevented if the input formula
does not contain disjunctions, in which case the footprints of all sub-formulae
are unique, i.e., singleton sets. The translated formulae are then again in a
theory decidable in $\NP$. \qed \end{proof}

\begin{theorem}
\label{theorem:complexity-lower-bound}
$\mathsf{SAT}(\pto, \gneg, \land, \lor, \star)$ is $\mathsf{PSPACE}$-complete.
\end{theorem}

\begin{proof}[sketch] Membership in $\PSPACE$ was proved in \cite{strong-sl} for
a more expressive fragment. For the hardness part, we build on the reduction
from QBF used in \cite{CC_results}. In this reduction, the boolean value of a
variable is represented by the corresponding SL variable being allocated (always
pointing to $\nil$ for simplicity). The fact that $x$ is false is expressed
using a negative points-to predicate stating that $x$ is not allocated. The
existential quantifier is expressed using the separating conjunction, and the
universal quantifier is obtained using the (unguarded) negation. (For details,
see \cite{CC_results}.)

We show that this reduction can be done without the unguarded negation and the
negative points-to assertion, using the guarded negation instead. The key
observation is that, for a QBF formula with variables $X$, we can express that
all variables in $X$ can have arbitrary boolean values as $\mathsf{arbitrary}[X]
\triangleq \bigstar_{x \in X} (x \pto \nil \lor \emp)$.  In the context of
variables $X$, we can then express negation as \mbox{$\neg F
\triangleq \mathsf{arbitrary}[X] \gneg F$} and the truth values of a variable~$x$
as \mbox{$\neg x \triangleq \mathsf{arbitrary}[X \setminus \{x\}]$} and $x \triangleq
\mathsf{arbitrary}[X] \star x \mapsto \nil$. The rest of the reduction then
easily follows~\cite{CC_results}. \qed \end{proof}

\vspace*{-5mm}
\section{Optimised Bound Computation}
\label{section:optimisations}
\vspace*{-0.5mm}

In many practical cases, the main source of complexity is the translation of
inductive predicates, which heavily depends on the possible lengths of paths
between locations. We now propose how to bound the length of these paths based
on the so-called \mbox{\textit{SL-graphs}} which are graph representations of constraints imposed by SL formulae. \mbox{SL-graphs} were originally used for representation and deciding of symbolic heaps with lists in \cite{sl_graphs}. Here, we use their
generalised form which captures must-relations holding in all models of a given
formula. Note that the nodes of the graphs are implicitly given by the domains of
the involved relations, which themselves can be viewed as edges.


\begin{definition}
An SL-graph of $\phi$ is a tuple $G[\phi] = (\musteq, \mustneq, (\mustpto{\f}, \mustpath{\f}, \muststar{\f})_{\f \in \Fields})$ where:
\begin{itemize}
    \item $\musteq \;\subseteq \vars(\phi) \times \vars(\phi)$ is an equivalence relation called must-equality,
    \item $\mustneq \;\subseteq \vars(\phi) \times \vars(\phi)$ is a symmetric relation called must-disequality,
    \item $\mustpto{\f} \;\subseteq \vars(\phi) \times \vars(\phi)$ is a must-$\f$-pointer relation,
    \item $\mustpath{\f} \;\subseteq \vars(\phi) \times \vars(\phi)$ is an irreflexive must-$\f$-path relation,
    \item $\muststar{\f} \subseteq \vars(\phi)^2 \times \vars(\phi)^2$ is a symmetric relation called must-$\f$-path-disjointness.
\end{itemize}
\end{definition}

\enlargethispage{2mm}

\noindent Except $\muststar{\f}$, the components of $G[\phi]$ represent atomic
formulae---equalities, disequalities, pointers, and paths (i.e., list
segments)---holding within all models of $\phi$. The fact that $(x_1, y_1)
\muststar{\f} (x_2, y_2)$ states that, in all models of $\phi$, the domains of
$\f$-paths from $x_1$ to $y_1$ and from $x_2$ to $y_2$ are disjoint.

To compute the SL-graph $G[\phi]$, we define some auxiliary notation. We define
$G_\emptyset$ to be an SL-graph where all the relations are empty. We write $G
\lhd \{x_i \bowtie_i y_i\}_{i \in I}$ to denote the SL-graph $G'$ which is the
same as $G$ with the elements $x_i \bowtie_i y_i$ for $i \in I$ added to the
corresponding relations. We use $\sqcup$ and $\sqcap$ as a component-wise union
and intersection of SL-graphs, respectively. We define the disjoint union of
SL-graphs as:
\begin{align*} G_1 \squplus G_2 &= (G_1 \sqcup G_2)\\
&\lhd \{ x \mustneq y \;|\; x
\in \alloc(G_1), y \in \alloc(G_2) \text{, and ($x$ is not $\nil$ or $y$ is not $\nil$})\}\\
&\lhd \{ e_1 \muststar{\f}
e_2 \;|\; \f \in \Fields, e_1 \in \paths_\f(G_1) \text{, and }  e_2 \in \paths_\f(G_2)\}.
\end{align*}

\noindent Here, $\paths_\f(G)$ is defined as $\mustpto{\f} \cup \mustpath{\f}$,
and the set of must-allocated variables is $\mathsf{alloc}(G) = \{x \;|\; \exists y, \f.\; x \mustpto{\f} y \text{ or } (x \mustpath{\f} y  \text{ and } x \mustneq y)\} \cup
\{\nil\}$ ($\nil$ is added for technical reasons). We further assume that all operations on SL-graphs ($\lhd$,
$\sqcup$, $\sqcap$, and~$\squplus$) preserve relational properties (symmetry,
transitivity, etc.) of the components of SL-graphs by computing the
corresponding closures after the operation is performed. We compute the SL-graph
$G[\phi]$ as follows.
\begin{align*}
    G[x = y] &= G_\emptyset \lhd \{ x \musteq y\}
    && G[x \pto \langle \f_i: f_i \rangle_{i \in I}] = G_\emptyset \lhd \{x \mustpto{\f_i} f_i\}_{i \in I}\\
    G[x \neq y] &= G_\emptyset \lhd \{ x \mustneq y\}
    && G[\lspred] = G_\emptyset \lhd \{ x \mustpath{\n} y\}\\
    G[\psi_1 \gneg \psi_2] &= G[\psi_1]
    && G[\dlspred] = G_\emptyset \lhd \{ x \mustpath{\n} y, x' \mustpath{\p} y'\}\\
    G[\psi_1 \land \psi_2] &= G[\psi_1] \sqcup G[\psi_2]
    && G[\nlspred] = G_\emptyset \lhd \{ x \mustpath{\n} z, x \mustpath{\t} y\}\\
    G[\psi_1 \lor \psi_2] &= G[\psi_1] \sqcap G[\psi_2]
    && G[\psi_1 \star \psi_2] = G[\psi_1] \squplus G[\psi_2]\vspace*{-1mm}
\end{align*}

\noindent Observe that we only approximate $\dls$ and $\nls$. After the construction is finished, we apply the following rules for matching of pointers and for detection of inconsistencies.

\begin{prooftree}
\AxiomC{$x_1 \mustpto{\f} y_1$\;\; $x_2 \mustpto{\f} y_2$\;\; $x_1 \musteq x_2$}
\RightLabel{($\mapsto$-match)\quad\quad}
\UnaryInfC{$y_1 \musteq y_2$}
\DisplayProof
\AxiomC{$x \musteq y$\;\; $x \mustneq y$}
\RightLabel{(contradiction)}
\UnaryInfC{$\phi$ is unsat}
\end{prooftree}

\vspace*{-1mm}\paragraph{Tighter location bounds.}

Using SL-graphs, we can slightly improve the location bound from Section \ref{section:small-models} by considering equivalence classes of $\musteq$  instead of individual variables (this can be also used to refine the later described path bound computation) and by defining $\chunksize{x} =1$ if $x$ is a must-pointers,
i.e., $x \mustpto{\f} y$ for some $\f$ and $y$.

\vspace*{-1mm}\paragraph{Path bounds.}

\enlargethispage{6mm}

We now fix an $\f$-path $\sigma$ from $x^S$ to $y$ and show how to compute an
interval $[\plower, \pupper]$ that gives bounds on its length.
The computation of the path bounds runs in two steps.
In the first step, we compute an initial bound $[\plower_{e}^0, \pupper_{e}^0]$
for each edge $e \in \paths_\f(G)$.
If $e$ is a pointer edge, its bound is given as $[1,1]$.
For a path edge $e = (a,b)$, we define $\plower_e^0 = 1$ if $a \mustneq b$ and
$0$ otherwise; while $\pupper_e^0$ is defined as \mbox{$\locsbound_S(\phi) - \sum_{v \in V} \chunksize{v}$} where $V = \{v \in \vars_S(\phi) \;|\; v \text{ is not } x \text{ and } \exists u.\; (v, u) \muststar{\f} (x, y)\}.$
\noindent This way, we exclude from the computation of the initial upper bound
the source $v$ of each path disjoint with $\sigma$ and all locations possibly
allocated in a chunk with the root $v$. Note that it can be the case that the
actual size of this chunk has a lesser size than $\chunksize{v}$, but this means
that we were too conservative when computing the global location bound and can
decrease the path bound by the same number anyway.

In the second phase, we compute the bounds of the path $\sigma$ using initial bounds from the first step. The computation is based on two
weighted directed graphs derived from the SL-graph $G$: $\gmax$ for the upper
bound and $\gmin$ for the lower bound (in both cases, the vertices are implicitly
given as $\vars(\phi)$, and the edge weight of an edge $e$ is given by $u_e^0$ and
$\ell_e^0$ computed in the previous step, respectively):
\begin{align*}
\gmax &= \{a \rightarrow b ~|~ (a, b) \in \mathsf{paths}_\f(G)\},\\
\gmin &= \{a \rightarrow b ~|~ 
    (a \mustpto{\f} b \text{ and } a \mustneq y) \text{ or }\\
    &\phantom{= \{a \rightarrow b ~|~}\; (a \mustpath{\f} b \text{ and } 
\exists w.\; \mathsf{nonempty}(y, w) \text{ and } (y, w) \muststar{\f} (a, b)\}.
\end{align*}
\noindent Here, the condition $\mathsf{nonempty}(y, w)$ states that a directed
SL-graph edge $(y, w)$ is non-empty which holds if either $y \mustpto{\f} w
\text{, or when } y \mustpath{\f} w \text { and } y \mustneq w$.

Intuitively, the upper bound $u$ is computed as the length of the shortest path
from $x$ to $y$ in~$\gmax$. Since $\f$-paths are uniquely determined, we know
that no path can be longer than the shortest one, and thus $u$ is indeed a
correct upper bound. The lower bound $\ell$ is computed as the length of the
longest path starting from $x$ (ending anywhere) in $\gmin$. By construction,
$\gmin$ contains only those edges for which one can prove that they cannot
contain $y$ in their domains. A path from $x$ of a length $\ell$ therefore implies
that $x$ cannot reach $y$ in less than $\ell$ steps, and thus $\ell$ is indeed a
correct lower bound.




\begin{figure*}[t]
    \centering
    \begin{minipage}{0.35\textwidth}
        \begin{subfigure}[b]{1\textwidth}
            \vspace*{-5mm}
            \centering
            \begin{tikzpicture}[
    node distance={0mm}, 
    thick, 
    main/.style = {draw, circle}]
    
    \tikzstyle{node}=[thick,draw=blue!75,fill=blue!20,minimum size=5mm,rounded corners=0.15cm]
    \tikzstyle{rnode}=[thick,draw=red!75,fill=red!20,minimum size=5mm,rounded corners=0.15cm]
    \tikzstyle{empty}=[]
    
    \node[node]  (1)                              {$a$};
    \node[node]  (2) [right of=1, xshift=12mm]    {$b$};
    \node[node]  (3) [right of=2, xshift=12mm]    {$c$};
    \node[node]  (4) [right of=3, xshift=12mm]    {$d$};

    \draw[-stealth', decorate, decoration={snake,amplitude=.3mm,segment length=2mm,post length=1mm}] 
        (1) to node [midway, below, empty] {$[0,2]$} (2);
    \draw [-stealth'] (2) to node [midway, below, empty] {$[1,1]$} (3);
    \draw[-stealth'] (3) to node [midway, below, empty] {$[1,1]$} (4);
    \draw[-stealth', bend right=70, decorate, decoration={snake,amplitude=.3mm,segment length=2mm,post length=1mm}] (4) to node [midway, above, empty] {$[0,2]$} (1);
    
    \draw[-] (2) [bend left=60, dotted] to node [midway, below, empty] {$\mustneq$} (3);

    \draw[-, draw, dotted, in=110, out=70] ($ (1) !.5! (2) $) to node [midway, above, empty] {$\muststar{\n}$} ($ (2) !.75! (4) $);
    

\end{tikzpicture}
            \caption{Fragment of SL-graph $G[\phi]$.}
            \label{fig:sl-graph}
        \end{subfigure}
    \end{minipage}
    \begin{minipage}{0.3\textwidth}
        \begin{subfigure}[b]{1\textwidth}
            \centering
            \vspace{8mm}
            \begin{tikzpicture}[
    node distance={0mm}, 
    thick, 
    main/.style = {draw, circle}]
    
    \tikzstyle{node}=[thick,draw=blue!75,fill=blue!20,minimum size=5mm,rounded corners=0.15cm]
    \tikzstyle{rnode}=[thick,draw=red!75,fill=red!20,minimum size=5mm,rounded corners=0.15cm]
    \tikzstyle{empty}=[]
    
    \node[node]  (1)                              {$a$};
    \node[node]  (2) [right of=1, xshift=11mm]    {$b$};
    \node[node]  (3) [right of=2, xshift=11mm]    {$c$};

    \draw[-stealth', preaction={draw,BurntOrange,-, double=BurntOrange, double distance=3\pgflinewidth,}] 
        (1) to node [midway, above, empty] {$0$} (2);
    \draw[-stealth', preaction={draw,BurntOrange,-, double=BurntOrange, double distance=3\pgflinewidth,}]
        (2) to node [midway, above, empty] {$1$} (3);

\end{tikzpicture}
            \vspace*{2.5mm}
            \caption{Graph $\gmin$.}
            \label{fig:path-min}
        \end{subfigure}
    \end{minipage}
    \begin{minipage}{0.3\textwidth}
        \begin{subfigure}[b]{1\textwidth}
            \centering
            \vspace{-1.5mm}
            \begin{tikzpicture}[
    node distance={0mm}, 
    thick, 
    main/.style = {draw, circle}]
    
    \tikzstyle{node}=[thick,draw=blue!75,fill=blue!20,minimum size=5mm,rounded corners=0.15cm]
    \tikzstyle{rnode}=[thick,draw=red!75,fill=red!20,minimum size=5mm,rounded corners=0.15cm]
    \tikzstyle{empty}=[]
    
    \node[node]  (1)                              {$a$};
    \node[node]  (2) [right of=1, xshift=11mm]    {$b$};
    \node[node]  (3) [right of=2, xshift=11mm]    {$c$};
    \node[node]  (4) [right of=3, xshift=11mm]    {$d$};

    \draw[-stealth', preaction={draw,BurntOrange,-, double=BurntOrange, double distance=3\pgflinewidth,}] 
        (1) to node [midway, above, empty] {$2$} (2);
    \draw[-stealth', preaction={draw,BurntOrange,-, double=BurntOrange, double distance=3\pgflinewidth,}] 
        (2) to node [midway, above, empty] {$1$} (3);
    \draw[-stealth']  
        (3) to node [midway, above, empty] {$1$} (4);
    \draw[-stealth', bend right=50] (4) to node [midway, above, empty] {$2$} (1);
\end{tikzpicture}
            \vspace{-4pt}
            \caption{Graph $\gmax$.}
            \label{fig:path-max}
        \end{subfigure}
    \end{minipage}

    \vspace*{-1mm} \caption{An illustration of the bound computation for
    the path $\sigma$ from $a$ to $c$ on a fragment of SL-graph of  $\phi
    \triangleq \bigl(\ls(a, b) \star b \pto c \star c \pto d \star \ls(d,
    a)\bigr) \gneg \bigl(\ls(a,c) \star \ls(c,a)\bigr)$. The highlighted edges
    denote the paths used to determine the bound $[1, 3]$.} \vspace*{-5mm}

\label{fig:path-bound}
\end{figure*}
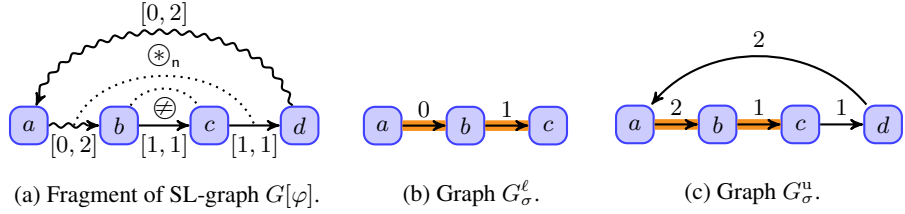

\vspace*{-1mm}\paragraph{Example.} We demonstrate the path bound computation in
Fig. \ref{fig:path-bound}, which shows a fragment of the SL-graph of a formula
$\phi$ (it shows only those $\muststar{\n}$ edges that are relevant in our
example) and the graphs $\gmin$ and $\gmax$ for the path $\sigma$ from $a$ to $c$.
We have that $\chunksize{b} = \chunksize{c} = 1$ and $\chunksize{a} =
\chunksize{d} = 2$. This gives us the location bound, which is 6. In the first
phase, we compute the initial bound $[0,2]$ for paths of the predicates $\ls(a,
b)$ and $\ls(d, a)$ because both of them are disjoint with all the other paths
in $G[\phi]$.
%
%
In the second phase, we get the bound for $\sigma$ equal to $[1, 3]$ instead of
the default \mbox{bound $[0, 6]$}.  \vspace{-3mm}

\section{Experimental Evaluation}
\label{section:experiments}
\vspace*{-1mm}

\enlargethispage{6mm}

We have implemented the proposed decision procedure in a new solver called
\mbox{\astral}{\footnote{\anonymize{Link}{\url{https://github.com/TDacik/Astral}}}}.
\astral is written in OCaml and can use multiple backend SMT solvers. With the
encoding presented in Section \ref{section:translation}, it can use either \cvc
supporting set theory directly \cite{cvc_sets} or \zzz supporting it by a
reduction to the extended theory of arrays \cite{z3_sets}. We have also
developed an alternative encoding in which both locations and location sets are
represented as bitvectors. The bitvector encoding differs only in expressing set
operations on the level of bitvectors with additional axioms ensuring that all
locations ``can fit'' into sets encoded by the bitvectors
\iftechreport
    (for details, see Appendix \ref{appendix:bitvectors}).
\else
    (for details, see \cite{tech_report}).
\fi
With the bitvector
encoding, a backend solver only needs to support theories of bitvectors and
arrays, which are both standard and supported by many other SMT solvers. 
Another advantage is that the
quantification on
%
%
bitvectors seems to perform significantly better than on
%
%
sets.

In our experiments, if we do not say explicitly which encoding and solver is
used, we use the bitvector encoding and \bitwuzla \cite{bitwuzla} as the
backend solver, which we found to be the best performing combination. We set a
limit for the method $\mathtt{SatEnum}$ to 64 footprints. If this limited is
exceeded, we dynamically switch to $\mathtt{SatQuantif}$. We use path quantifiers
when the path bound is at most half of the domain bound. These are design
choices that can be revisited in the future.

All experiments were run on a machine with 2.5 GHz Intel Core i5-7300HQ CPU and 16 GiB RAM, running Ubuntu 18.04. The timeout was set to 60 s and the memory limit to 1 GB. Our experiments were conducted using \benchexec~\cite{benchexec}, a framework for reliable benchmarking.

\begin{table*}[!t]

    \caption{Experimental results for formulae from SL-COMP. The columns are:
    solved instances (OK), out of time/memory (RO), instances on which \astral wins---\astral can solve it and the other solver not or \astral solves it faster (WIN), instances solved in the time limits of $0.1$ s and $1$ s, and the total time for solved \mbox{instances in seconds.}}

    \begin{subfigure}[b]{1\textwidth}
        \centering
        \caption{Results for the category QF\_SHLS\_ENTL.}
        \begin{tabular}{@{}lrrrrrrcrrrrrr@{}}
\toprule

& \multicolumn{6}{c}{Verification conditions (86)} & \phantom{abc}& \multicolumn{6}{c}{bolognesa+clones (210)}\\

\cmidrule{2-7} \cmidrule{9-14}

Solver & OK & \hspace{1pt} RO & WIN & $<\!\! 0.1$ & \hspace{1pt} $\leq\!\! 1$ & Total time && 
         OK & \hspace{1pt} RO & WIN & $<\!\! 0.1$ & \hspace{1pt} $\leq\!\! 1$ & Total time\\
         
\midrule
\rowcolor{GreenYellow}
\textsc{Astral} & 86 & 0 & - & 84 & 86 & 4.62 && 210 & 0 & - & 68 & 169 & 202.91\\
\textsc{Astral}$^\star$ & 86 & 0 & 42 & 83 & 86 & 4.64 && 195 & 15 & 88 & 64 & 150 & 408.48\\
\rowcolor[gray]{0.9}
\textsc{GRASShopper} & 86 & 0 & 70 & 52 & 86 & 8.65 && 203 & 7 & 148 & 60 & 87 & 1229.35\\
\textsc{S2S} & 86 & 0 & 5 & 86 & 86 & 2.08 && 210 & 0 & 3 & 203 & 210 & 8.18\\
\rowcolor[gray]{0.9}
\textsc{Sloth} & 64 & 3 & 86 & 0 & 28 & 235.28 && 70 & 140 & 210 & 0 & 50 & 149.42\\

\bottomrule
\end{tabular}
        \label{table:experiments1}
    \end{subfigure}
    \begin{subfigure}[b]{1\textwidth}
        \centering
        \caption{Results for a subset of the category QF\_SHLID\_ENTL.}
        \begin{tabular}{@{}lrrrrrrcrrrrrr@{}}
\toprule

& \multicolumn{6}{c}{Doubly-linked lists (17)} & \phantom{abc}& \multicolumn{6}{c}{Nested singly-linked lists (19)}\\

\cmidrule{2-7} \cmidrule{9-14}

Solver & OK & \hspace{1pt} RO & WIN & $<\!\! 0.1$ & \hspace{1pt} $\leq\!\! 1$ & Total time && 
         OK & \hspace{1pt} RO & WIN & $<\!\! 0.1$ & \hspace{1pt} $\leq\!\! 1$ & Total time\\

\midrule
\rowcolor{GreenYellow}
\textsc{Astral} & 17 & 0 & - & 11 & 17 & 2.72 && 19 & 0 & - & 3 & 9 & 86.93\\
\textsc{GRASShopper} & 17 & 0 & 16 & 3 & 15 & 7.53 && - & - & - & - & - & -\\
\rowcolor[gray]{0.9}
\textsc{Harrsh} & 17 & 0 & 17 & 0 & 0 & 95.18 && 14 & 5 & 18 & 0 & 0 & 183.01\\
\textsc{S2S} & 17 & 0 & 0 & 17 & 17 & 0.15 && 19 & 0 & 0 & 19 & 19 & 0.43\\
\rowcolor[gray]{0.9}
\textsc{Songbird} & 11 & 5 & 14 & 5 & 9 & 13.39 && 11 & 5 & 8 & 4 & 11 & 1.38\\
\bottomrule
\end{tabular}
        \label{table:experiments2}
    \end{subfigure}

    \vspace*{-6mm}
\end{table*}

\vspace*{-4mm}\subsection{Entailments of Symbolic Heaps}\vspace*{-1mm}

In the first part of our evaluation, we focus on formulae from the symbolic heap
fragment which is frequently used by verification tools and for which there
exist many dedicated solvers. We therefore do not expect to outperform the best
existing tools but rather to obtain a comparison with other translation-based
decision procedures.

In Table \ref{table:experiments1}, we provide results for the category
QF\_SHLID\_ENTL (entailments with SLLs). We divide the category into two
subsets: verification conditions (which are simpler) and more complex
artificially generated formulae ``\textit{bolognesa}" and ``\textit{clones}"
from \cite{SL-superposition}. During the experiments, we found out that several
``cloned'' entailments contain root variables on the right-hand side of the
entailment that do not appear on the left-hand side, making the entailment
trivially invalid when its left-hand side is satisfiable. For a few hard clone
instances, this makes a problem for \astral as it cannot use the path bound
computation as such roots do not appear in the SL-graph. We have therefore
implemented a heuristic that detects  entailments $\phi \models \psi$ that can
be reduced to satisfiability of $\phi$. Since this is a benchmark-specific heuristic, we present
also the version without this heuristic (\astral$\!^\star$) in Table
\ref{table:experiments1}. The optimised version of \mbox{\astral} is able to solve all
the formulae being faster than other translation-based solvers
\mbox{\grasshopper}\footnote{Since \grasshopper is not an solver but a verification
tool, we encode the entailment checking as a verification of an empty program.}
and \sloth. For illustration, the table further contains the second best solver
in the latest edition of SL-COMP, \sss\footnote{We had technical issues running
the winner \asterix \cite{sl_mod_theories}. The difference between those tools
is, however, negligible.}.

In Table \ref{table:experiments2}, we provide results for a subset of the
category QF\_SHLID\_ENTL (entailments with linear inductive definitions from
which we selected DLLs and NLLs) for \astral and three best-performing solvers
competing in the latest edition of SL-COMP---\sss, \songbird (in the version with automated lemma synthesis called \sls), and \harrsh. We
also include \grasshopper which supports DLLs only. Except \sss which solves
almost all formulae virtually immediately, \astral is the only one able to solve
all the formulae in the given time limit.

\vspace*{-3mm}
\subsection{Experiments on Formulae Outside of the Symbolic Heap Fragment}
\vspace*{-1mm}

\enlargethispage{6mm}

For formulae outside of the symbolic heap fragment and its top-level boolean
closure, there are currently no existing benchmarks. For now, we therefore limit
ourselves to randomly generated but extensive sets of formulae. In the future,
we would like to develop a program analyser using symbolic execution over BSL
and make more careful experiments on realistic formulae.

We first focus on the fragment with guarded negations but without inductive
predicates, on which we can compare \astral with \cvc. We have prepared a set of
1000 entailments of the form $\phi \models \psi$ which are generated as random
binary trees with depth~8 over 8 variables with the only atoms being pointer
assertions. To reduce the number of trivial instances, we only generated
formulae for which $\vars(\psi) \subseteq \vars(\phi)$ and \mbox{\astral} cannot deduce
contradiction from their SL-graphs. To avoid any suspicion that the difference
is caused by better performance of the backend solver rather than the design of
our translation, we used \astral with the $\cvc$ backend and direct set encoding (with \bitwuzla and bitvector encoding, our
results would be even better). The results are given in
Fig.~\ref{fig:random_cvc5} and suggest that our treatment of guarded negations
really brings a better performance---\astral can solve all the instances and
almost all of them under 10 seconds. On the other hand, \cvc timed out in 61
cases and is usually slower than \astral, in particular on satisfiable formulae
which represent invalid entailments.

In the second experiment, we compared our solver with \grasshopper on the
fragment which it supports, i.e., arbitrary nesting of conjunctions and
disjunctions. We again generated 1000 entailments, this time with depth 6, 6
variables and with atoms being singly-linked lists (with 20~\% probability) or
pointer-assertions. The results are given in Fig.~\ref{fig:random_grasshopper}.
\astral ran out of memory in 5 cases, and \grasshopper timed out in 10 cases. In
summary, \astral is faster on more than 80~\% of the formulae with an almost 3
times lesser running time.

Finally, to illustrate that \astral can indeed handle formulae out of the
fragments of all the other mentioned tools, we apply it on an entailment query
that involves the formula mentioned at the end of the introduction: $((\ls(x,y) \gneg (\ls(x,z) \star \ls(z,y))) \star y \pto z) \models \ls(x,z),$
converted to
an unsatisfiability query. \astral resolves the query in 0.12 s. Note that
without the requirement $\neg (\ls(x,z) \star \ls(z,y))$, the entailment does not
hold as a cycle may be closed in the heap.

\begin{figure}[!t]
    \centering
    \begin{subfigure}[b]{0.48\textwidth}
        \vspace*{-10mm}
        \includegraphics[width=\textwidth]{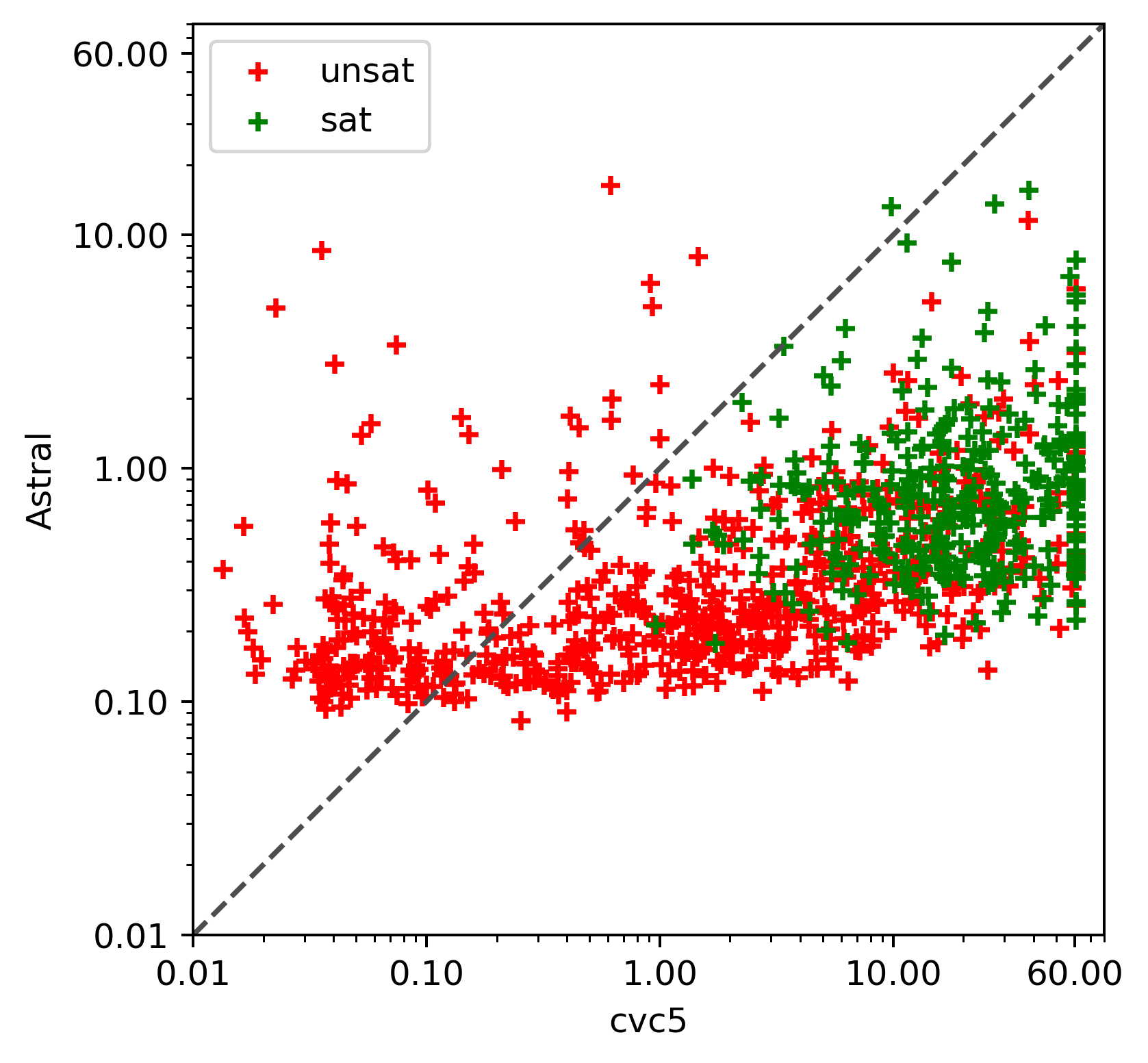}
        \vspace*{-6mm}
        \caption{Comparison with \cvc}
        \label{fig:random_cvc5}
    \end{subfigure}
    \begin{subfigure}[b]{0.48\textwidth}
        \centering
        \includegraphics[width=\textwidth]{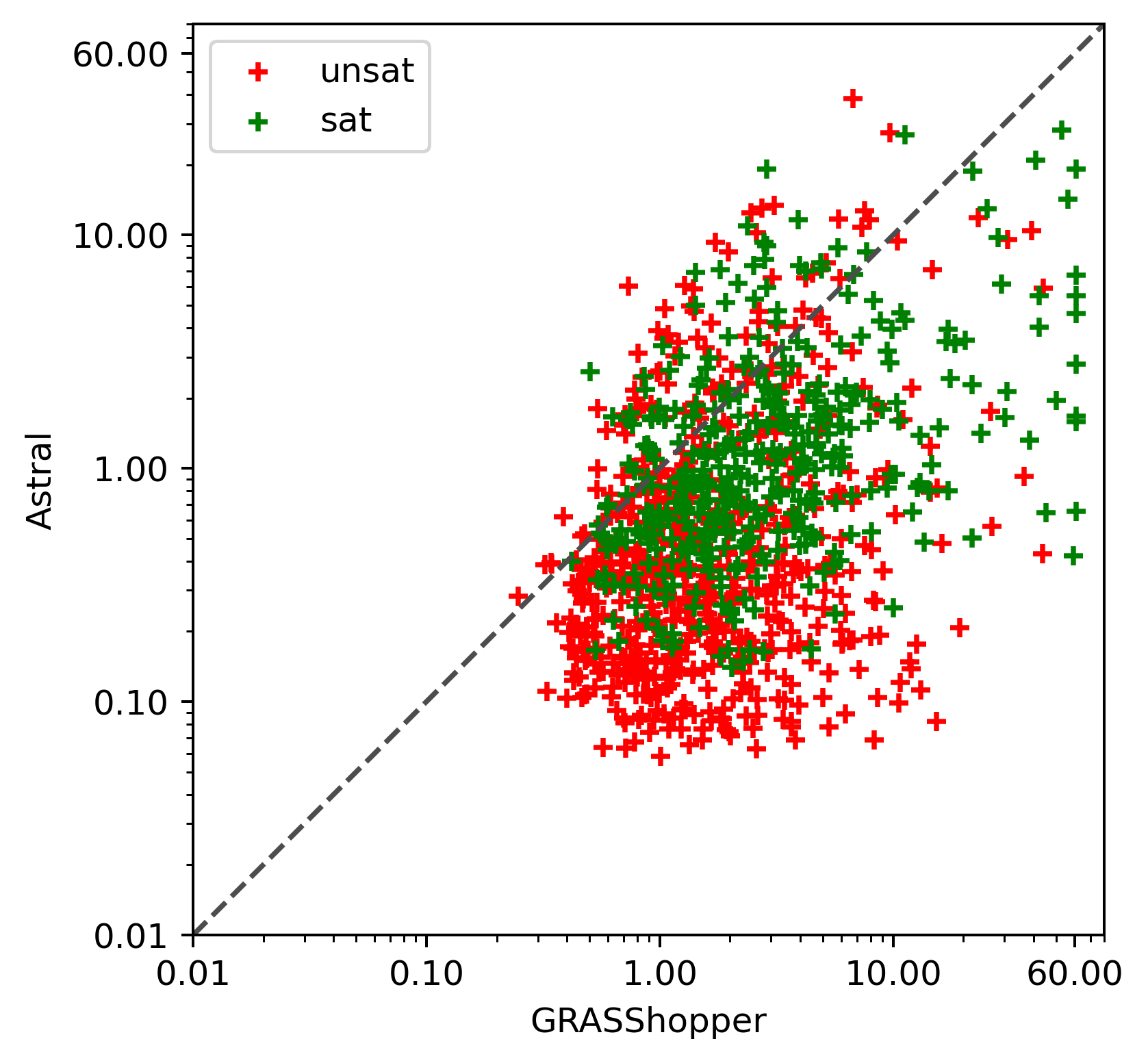}
        \vspace*{-6mm}
        \caption{Comparison with \grasshopper}
        \label{fig:random_grasshopper}
    \end{subfigure}

    \caption{A comparison of \astral with \cvc and \grasshopper on randomly
    generated formulae. Times are in seconds, axes are logarithmic. The timeout was \mbox{set to 60 s.}}\vspace*{-5mm}

\end{figure}

\enlargethispage{6mm}

\vspace*{-4mm}
\section{Conclusions and Future Work}
\vspace*{-2mm}
\label{section:conclusions}

We have presented a novel decision procedure based on a small-model property and
translation to SMT. Our experiments have shown very promising results,
especially for formulae with rich boolean structure for which our decision
procedure outperforms other approaches (apart from being able to solve more
formulae).


In the future, we would like to extend our approach with some class of
user-defined inductive predicates, with more complex spatial connectives such
as septractions and/or magic wands, consider a lazy and/or interactive
translation instead of the current eager approach, and try \astral within some
SL-based program analyser.


\newpage


%
%
%
\bibliographystyle{splncs04}
\bibliography{references}

\begin{thebibliography}{10}
\providecommand{\url}[1]{\texttt{#1}}
\providecommand{\urlprefix}{URL }
\providecommand{\doi}[1]{https://doi.org/#1}

\bibitem{cvc_sets}
Bansal, K., Barrett, C., Reynolds, A., Tinelli, C.: {A New Decision Procedure for Finite Sets and Cardinality Constraints in SMT}. In: IJCAR (2017)

\bibitem{quantitative_sl}
Batz, K., Fesefeldt, I., Jansen, M., Katoen, J.P., Ke{\ss}ler, F., Matheja, C., Noll, T.: {Foundations for Entailment Checking in Quantitative Separation Logic}. In: ESOP (2022)

\bibitem{decidable_sl}
Berdine, J., Calcagno, C., O'Hearn, P.W.: {A Decidable Fragment of Separation Logic}. In: FSTTCS 2004. LNCS, vol.~3328 (2004)

\bibitem{benchexec}
Beyer, D., L{\"o}we, S., Wendler, P.: {Reliable Benchmarking: Requirements and Solutions}. International Journal on Software Tools for Technology Transfer  \textbf{21} (2017)

\bibitem{Brotherston2012AGC}
Brotherston, J., Gorogiannis, N., Petersen, R.L.: {A Generic Cyclic Theorem Prover}. In: APLAS. LNCS, vol.~7705 (2012)

\bibitem{Dino-BiAbd11}
Calcagno, C., Distefano, D., O'Hearn, P., Yang, H.: {Compositional Shape Analysis by Means of Bi-Abduction}. Journal of the ACM  \textbf{58}(6) (2011)

\bibitem{CC_results}
Calcagno, C., Yang, H., O'Hearn, P.W.: {Computability and Complexity Results for a Spatial Assertion Language for Data Structures}. In: FST TCS (2001)

\bibitem{sl_graphs}
Cook, B., Haase, C., Ouaknine, J., Parkinson, M., Worrell, J.: {Tractable Reasoning in a Fragment of Separation Logic}. In: CONCUR. LNCS, vol.~3901 (2011)

\bibitem{bsr_sl}
Echenim, M., Iosif, R., Peltier, N.: {The Bernays-Schönfinkel-Ramsey Class of Separation Logic with Uninterpreted Predicates}. ACM Transactions on Computational Logic  \textbf{21} (2019)

\bibitem{spen}
Enea, C., Leng\'{a}l, O., Sighireanu, M., Vojnar, T.: {Compositional Entailment Checking for a Fragment of Separation Logic}. In: APLAS (2014)

\bibitem{Broom}
Hol{\'{\i}}k, L., Peringer, P., Rogalewicz, A., \v{S}okov{\'{a}}, V., Vojnar, T., Zuleger, F.: Low-level bi-abduction. In: {ECOOP} 2022. LIPIcs, vol.~222, pp. 19:1--19:30 (2022)

\bibitem{SL_tree_automata}
Iosif, R., Rogalewicz, A., Vojnar, T.: {Deciding Entailments in Inductive Separation Logic with Tree Automata}. In: ATVA (2014)

\bibitem{conf/concur/IosifZ23}
Iosif, R., Zuleger, F.: Expressiveness results for an inductive logic of separated relations. In: P{\'{e}}rez, G.A., Raskin, J. (eds.) {CONCUR}. LIPIcs, vol.~279, pp. 20:1--20:20 (2023). \doi{10.4230/LIPICS.CONCUR.2023.20}, \url{https://doi.org/10.4230/LIPIcs.CONCUR.2023.20}

\bibitem{OHearn:BI:01}
Ishtiaq, S., O'Hearn, P.: {Separation and Information Hiding}. In: Proc. of POPL'01. ACM (2001)

\bibitem{sl_data}
Katelaan, J., Jovanovic, D., Weissenbacher, G.: {A Separation Logic with Data: Small Models and Automation}. In: IJCAR (2018)

\bibitem{harrsh}
Katelaan, J., Matheja, C., Noll, T., Zuleger, F.: {Harrsh: A Tool for Unied Reasoning about Symbolic-Heap Separation Logic}. In: LPAR-22 Workshop and Short Paper Proceedings. vol.~9 (2018)

\bibitem{SecOrderBiAbd14}
Le, Q.L., Gherghina, C., Qin, S., Chin, W.N.: {Shape Analysis via Second-Order Bi-Abduction}. In: Proc. of CAV'14. LNCS, vol.~8559. Springer (2014)

\bibitem{S2S}
Le, Q.L.: {Compositional Satisfiability Solving in Separation Logic}. In: VMCAI. LNCS, vol. 12597 (2021)

\bibitem{S2S2}
Le, Q.L., Le, X.B.D.: {An Efficient Cyclic Entailment Procedure in a Fragment of Separation Logic}. In: FoSSaCS (2023)

\bibitem{Guarded-SL}
Matheja, C., Pagel, J., Zuleger, F.: {A Decision Procedure for Guarded Separation Logic Complete Entailment Checking for Separation Logic with Inductive Definitions}. ACM Trans. Comput. Logic  \textbf{24}(1) (2023)

\bibitem{z3_sets}
de~Moura, L., Bjørner, N.: {Generalized, efficient array decision procedures}. In: FMCAD (2009)

\bibitem{SL-superposition}
Navarro~P\'{e}rez, J.A., Rybalchenko, A.: {Separation Logic + Superposition Calculus = Heap Theorem Prover}. In: PLDI (2011)

\bibitem{sl_mod_theories}
Navarro~Pérez, J.A., Rybalchenko, A.: {Separation Logic Modulo Theories}. In: APLAS. LNCS, vol.~8301 (2013)

\bibitem{bitwuzla}
Niemetz, A., Preiner, M.: Bitwuzla. In: CAV. LNCS, vol. 13965 (2023)

\bibitem{strong-sl}
Pagel, J., Zuleger, F.: Strong-separation logic. {ACM} Trans. Program. Lang. Syst.  \textbf{44}(3),  16:1--16:40 (2022). \doi{10.1145/3498847}, \url{https://doi.org/10.1145/3498847}

\bibitem{automating_sl}
Piskac, R., Wies, T., Zufferey, D.: {Automating Separation Logic Using SMT}. In: CAV (2013)

\bibitem{automating_sl_trees}
Piskac, R., Wies, T., Zufferey, D.: {Automating Separation Logic with Trees and Data}. In: CAV (2014)

\bibitem{cvc_sl}
Reynolds, A., Iosif, R., King, T.: {A Decision Procedure for Separation Logic in SMT}. In: ATVA (2016)

\bibitem{SL}
Reynolds, J.: {Separation Logic: A Logic for Shared Mutable Data Structures}. In: Proceedings 17th Annual IEEE Symposium on Logic in Computer Science (2002)

\bibitem{GilPartI20}
Santos, J., Maksimovic, P., Ayoun, S.E., Gardner, P.: {Gillian, Part I: A Multi-Language Platform for Symbolic Execution}. In: Proc. of PLDI'20. ACM (2020)

\bibitem{DBLP:journals/sttt/SummersM20}
Summers, A.J., M{\"{u}}ller, P.: Automating deductive verification for weak-memory programs (extended version). Int. J. Softw. Tools Technol. Transf.  \textbf{22}(6),  709--728 (2020). \doi{10.1007/s10009-020-00559-y}

\bibitem{10.1145/3158097}
Ta, Q.T., Le, T.C., Khoo, S.C., Chin, W.N.: {Automated Lemma Synthesis in Symbolic-Heap Separation Logic}. In: POPL (2018)

\bibitem{InvaderCAV08}
Yang, H., Lee, O., Berdine, J., Calcagno, C., Cook, B., Distefano, D., O'Hearn, P.: {Scalable Shape Analysis for Systems Code}. In: Proc. of CAV'08. LNCS, vol.~5123. Springer (2008)

\end{thebibliography}


\iftechreport
    \newpage
    \appendix
    \section{Proof of Small-Model Property}
\enlargethispage{6mm}

This section presents proofs omitted in Section \ref{section:small-models}. We first introduce some additional notation. If a stack $s$ is clear from the context, we call a heap $h$ positive (atomic) if $(s,h)$ is positive (atomic). We say that a model $(s,h)$ is a canonical model of a formula~$\phi$ if $(s,h) \models \phi$ and $\reduction^X\!(s,h) = (s,h)$ for $X = \vars(\phi)$.

\begin{proof}[Lemma \ref{lemma:atom-classification}]
Let $(s,h)$ be an atomic model with $(s,h) \models \phi$. Consequently $h \neq \emptyset$ by the definition of atomicity. We proceed by structural induction on $\phi$. Cases of pure formulae cannot happen since $h \neq \emptyset$. Cases of pointer assertions are trivial as they always meet precisely the condition (1). Cases of inductive predicates are as follows:

\begin{itemize}
    \item If $(s,h) \models \ls(x,y)$, then $h$ is either a single pointer $x \mapsto y$ that satisfies~(1), or otherwise it satisfies (2).
    \item If $(s,h) \models \dlspred$, then $h$ is either a single pointer $x \mapsto \cdls{y}{y'}$ that satisfies~(1), or otherwise it satisfies (3). 
        The case when $h$ consists of two pointers $x \mapsto \cdls{x'}{y'}$ and $x' \mapsto \cdls{y}{x}$ cannot happen because it contradicts the assumption that $(s,h)$ is atomic.
    \item If $(s,h) \models \nls(x,y,z)$, then $h$ is either a single pointer $x \mapsto \cnls{z}{y}$ which satisfies~(1), or otherwise it satisfies (4).
\end{itemize}

If $(s, h) \models \psi_1 \star \psi_2$, then there exist disjoint heaps $h_1$ and $h_2$ with $h = h_1 \uplus h_2 \neq \bot$ satisfying $\psi_1$ and $\psi_2$, respectively. Since $(s,h)$ is atomic, exactly one of the heaps $h_1$ and $h_2$ needs to be empty. Assume w.l.o.g. that $h_1 = \emptyset$ and $h = h_2$. Hence $(s, h) \models \psi_2$ and we obtain the claim using IH. The cases of boolean connectives follow directly from IH. \qed
\end{proof}

To prove Lemma \ref{lemma:chunk-decomposition}, we introduce a definition of \textit{split point}. Intuitively, a split point is a location that can be used to split a model of an inductive predicate into two non-empty models of the same predicate type. If a model contains a split point, then it is not atomic and consequently also not a chunk.

\begin{definition}[Split point]
Let $(s,h)$ be a model such that $h \neq \emptyset$ and $(s,h) \models \pi(x, y)$ for an inductive predicate $\pi(x, y)$. Let $\ell \in \dom(h)$ be a location such that $\ell~\in~\img(s)$. We say that $\ell$ is a \emph{split point} of $(s,h)$ if $s(x) \neq \ell$ and  there exist $x', y', z$ such that one of the following conditions holds:
\begin{itemize}
    \item $\pi \triangleq \lspred$,
    \item $\pi \triangleq \nlspred$ and $h(\ell, \t) \neq \bot$,
    \item $\pi \triangleq \dlspred$ and $h(\ell, \p) \in \img(s)$.
\end{itemize}
\end{definition}

\begin{lemma}
Let $(s,h)$ be a model such that $(s,h) \models \pi$ for some inductive predicate~$\pi$. If $(s,h)$ contains a split point, then it is not atomic.
\end{lemma}

\begin{proof} Let $\ell$ be a split point of $(s,h)$ and let $p$ be some variable such that $s(p) = \ell$. We will show that $(s,h)$ is not atomic as it can be split into two non-empty positive sub-heaps satisfying formulae $\psi_1$ and~$\psi_2$ given as follows:

\begin{itemize}
    \item If $\pi \triangleq \lspred$, then $\psi_1 \triangleq \ls(x, p)$ and $\psi_2 \triangleq \ls(p, y)$.

    \item If $\pi \triangleq \nlspred$, then $\psi_1 \triangleq \nls(x, p, z)$ and $\psi_2 \triangleq \nls(p, y, z)$.

    \item If $\pi \triangleq \dlspred$, then $\psi_1 \triangleq \dls(x, p, p', y')$ and $\psi_2 \triangleq \dls(p, y, x', p')$, where~$p'$ is some variable such that $s(p') = h(\ell, \p)$ whose existence is guaranteed by the definition of split points.

\end{itemize}
In all the cases, non-emptiness of the model of $\psi_1$ follows from the fact that $s(x) \neq s(p)$ by the definition of split points; and non-emptiness of the model of $\psi_2$ follows from the fact that $s(y) \neq s(p)$ because $s(p) \in \dom(h)$, but $s(y) \not \in \dom(h)$. \qed
\end{proof}

\enlargethispage{10mm}
We prove Lemma \ref{lemma:chunk-decomposition} by proving the following stronger claim.

\begin{lemma}
\label{lemma:chunk-decomposition_aux}
Let $(s,h)$ be a positive model. The the following claims hold:
\begin{enumerate}
    \item $h = \biguplus \chunks(s,h)$
    \item If $h_1$ and $h_2$ are non-empty positive sub-heaps such that $h = h_1 \uplus h_2$, then $\chunks(s,h) = \chunks(s,h_1) \uplus \chunks(s,h_2)$.
\end{enumerate}
\end{lemma}

\begin{proof}
By induction on $|h|$.
If $|h| = 0$, then both claims hold as $\chunks(s,h) = \emptyset$ and $h$~cannot be split into non-empty sub-heaps.
If $|h| = n+1$ and $h$ is atomic, the claims also hold because $\chunks(s,h) = \{h\}$ and $h$ cannot be further decomposed into positive sub-heaps.
Otherwise, if $h$ is not atomic, then there exist non-empty positive heaps $h_1$ and $h_2$ such that $h = h_1 \uplus h_2$. Since $h_1$ is non-empty, we have that $|h_2| \leq n$ and vice versa. Therefore we can apply IH to conclude that \mbox{$h_i = \biguplus \chunks(s,h_i)$ for $i = 1, 2$.}

Since $h = h_1 \uplus h_2$, we directly have that $h = \biguplus \bigl( \chunks(s,h_1) \uplus \chunks(s,h_2) \bigr)$. It remains to show that $\chunks(s,h) = \chunks(s,h_1) \uplus \chunks(s,h_2)$ which directly implies both (1) and (2). Namely, we need to show that (a) there does not exist chunk $c$ of $(s,h)$ such that $c \not \in \chunks(s,h_i)$ for $i = 1, 2$ and (b) each element of $\chunks(s,h_i)$ is also member of $\chunks(s,h)$ for $i = 1, 2$.

\begin{itemize}
    \item[(a)] By contradiction. Assume that there exists such a chunk $c$. It holds that $c$ cannot be a pointer-chunk, because then it would be chunk in either $(s,h_1)$ or $(s,h_2)$. Therefore $c$ is a $\pi$-chunk for some inductive predicate $\pi(x, y)$. Since $c$ is non-empty, we can w.l.o.g. assume that the location $s(x)$ is allocated in $h_1$. As $h_1$ is fully decomposed into its chunks by IH, there exists a chunk $c' \in \chunks(s,h_1)$ with root $x'$ such that $s(x) \in \dom(c')$. We proceed by case distinction on types of $c$ and~$c'$:

    \begin{itemize}
        \item If $c'$ is a-pointer chunk satisfying $x' \mapsto y'$ for some $y'$ which is $n$, \mbox{$\cdls{n}{\wildcard}$}, or $\cnls{\wildcard}{t}$, then $n$ or $t$ is a split point of the chunk $c$ which is a contradiction.

        \item If $c'$ is $\pi(x', y')$-chunk for some $y'$ and $\pi \in \{\dls, \nls\}$, then $c$ also needs to be $\pi$-chunk. If $s(x) = s(x')$, then either $c = c'$ which is a contradiction, or $c \subseteq c'$ which means that $y$ is a split point of $c'$, or $c' \subseteq c$ which means that $y'$ is a split point of $c$---both possibilities lead to a contradiction. If $s(x) \neq s(x')$, then $s(x)$ is a split point of $c'$ which again leads to a contradiction.

        \item If both $c$ and $c'$ are $\ls$-chunks we can reach contradiction by the same reasoning as in the previous step.

        \item The only remaining case is when $c$ is $\ls$-chunk and $(s, c) \models \nls(x', y', z')$ for some $y'$ and $z'$. It holds that $c \subset c'$ because otherwise $s(z')$ would be a split point of $c$ or $c = c'$, both leading to a contradiction. From $c \subset c'$ and $c, c' \subseteq h$, we have that $c \not \in \chunks(s,h)$ as it is not maximal, which is a contradiction.
    \end{itemize}

    \item[(b)] By contradiction. Assume w.l.o.g. that there exists $c$ such that $c \in \chunks(s,h_1)$ and $c \not \in \chunks(s,h)$. Consequently, there must exists $c' \subseteq h$ such that $c \subset c'$. It holds that $c' \not \in \chunks(s,h_1)$ because otherwise $c$ would not be a chunk of $(s, h_1)$. As we have already shown in (a), an existence of such $c'$ leads to a contradiction. \qed
\end{itemize}
\end{proof}

\newpage
\enlargethispage{4mm}

To prove Theorem \ref{theorem:reduction}, we need to show that for arbitrary stack $s$, $\reduction^s$ is homomorphism on the separation algebra of all positive heaps (w.r.t. the stack $s$).

\begin{lemma}[$\reduction^s$ is homomorphism]
\label{lemma:reduction-homomorphism}
Let $s$ be a stack and let $h_1$ and $h_2$ be positive heaps. If $h_1$ and $h_2$ are disjoint, then the following claims hold:
\begin{enumerate}
    \item $\reduce{s}{h_1}$ and~$\reduce{s}{h_2}$ are disjoint,
    \item $\reduce{s}{h_1} \;\uplus \reduce{s}{h_2} =\;\reduce{s}({h_1 \uplus h_2})$.
\end{enumerate}
\end{lemma}

\begin{proof} $\;$
\begin{enumerate}
    \item Directly from the definition of $\reduction^s$, we have that $\dom(\reduce{s}{h_i}) \subseteq \dom(h_i)$ for $i = 1, 2$. Hence, if $\dom(h_1) \cap \dom(h_2) = \emptyset$, then also $\dom(\reduce{s}{h_1}) \cap \dom(\reduce{s}{h_2)} = \emptyset$.
    \vspace{1mm}

    \item Let $\chunks(s, h_1) = \{c_{1,1}, \ldots, c_{1,m}\}$ and $\chunks(s, h_2) = \{c_{2,1}, \ldots, c_{2,n}\}$.
    \begin{align*}
        &\reduce{s}{h_1} \;\uplus\; \reduce{s}{h_2}\\
            & =\; \reduce{s}{(c_{1,1} \uplus \cdots \uplus c_{1,m})} \;\uplus\; \reduce{s}{(c_{2,1} \uplus \cdots \uplus c_{2,n})}
            && [\text{Lemma \ref{lemma:chunk-decomposition_aux}}]\\
            & =\; \reduce{s}{c_{1,1}} \uplus \cdots \uplus \reduce{s}{c_{1,m}} \;\uplus\; \reduce{s}{c_{2,1}} \uplus \cdots \uplus \reduce{s}{c_{2,n}}
            && [\text{Definition of $\reduction^s$}]\\
            & =\; \reduce{s}{(c_{1,1} \uplus \cdots \uplus c_{1,m} \;\uplus\; c_{2,1} \uplus \cdots \uplus c_{2,n})}
            && [\text{Definition of $\reduction^s$, Lemma \ref{lemma:chunk-decomposition_aux}}]\\
            & =\; \reduce{s}{(h_1 \uplus h_2)}
            && [\text{Lemma \ref{lemma:chunk-decomposition_aux}}]
    \end{align*}
    \noindent Notice that for the third equality to hold, we need Lemma \ref{lemma:chunk-decomposition_aux} which guarantees that $\chunks(s,h_1 \uplus h_2) = \chunks(s,h_1) \uplus \chunks(s,h_2)$. \qed
\end{enumerate}
\end{proof}

\begin{corollary}
\label{corollary:stack-restriction}
For a positive model $(s,h)$, it holds that $(s,h) \models \phi$ iff $(s|_{\vars(\phi)}, h) \models \phi$.
\end{corollary}

\begin{proof}
By straightforward structural induction on $\phi$ using the fact that whether it holds that $(s,h)~\models~\phi$ does not depend on variables outside of $\vars(\phi)$.
\end{proof}

\begin{proof}[Theorem \ref{theorem:reduction}]
Let $s' = s|_{\vars(\phi)}$ and $h' = \reduce{s'}{h}$. From Corollary $\ref{corollary:stack-restriction}$ we have that $(s, h) \models \phi$ iff $(s', h) \models \phi$. It remains to show that $(s', h) \models \phi$ iff $(s', h') \models \phi$. We proceed by structural induction on $\phi$.
\begin{itemize}
    \item \textit{Base cases.} If $\phi$ is a pure formula or a pointer assertion, then $h$ is empty or a singleton set, respectively, and the claim holds as in both cases $h' = h$. If $\phi$ is an inductive predicate $\pi(x, y)$, then both $(s', h)$ and $(s', h')$ consist of a single chunk because their only allocated variable is $x$ and they therefore cannot be split into two non-empty positive models (one of splits would need to contain an allocated root different from $x$). The implication $(\Rightarrow)$ then holds trivially from the definition of the reduction. The implication $(\Leftarrow)$ is proved by rule inversion on the definition of the reduction---if $h'$ is $\pi$-chunk, then $h$ was also a $\pi$-chunk satisfying the same predicate $\pi$.

    \item \textit{Induction steps.} The cases of boolean connectives follow directly from IH. The case of separating conjunction is more involved:

    \begin{itemize}[leftmargin=2.5em,labelindent=1em]
        \item[$(\Rightarrow)$]
        By the assumption, there exist disjoint heaps $h_1$, $h_2$ such that $h = h_1 \uplus h_2$ and $(s',h_i) \models \psi_i$ for $i=1,2$. From IH, we have that $(s', \reduce{s'}{h_i}) \models \psi_i$ for $i = 1,2$ and from Lemma \ref{lemma:reduction-homomorphism}, we have that $\reduce{s'}{h_1}$ and $\reduce{s'}{h_2}$ are also disjoint. Hence, $(s', \reduce{s'}{h_1} \uplus \reduce{s'}{h_2}) \models \phi$. From Lemma \ref{lemma:reduction-homomorphism}, we also have that $\reduce{s'}{h_1} \uplus \reduce{s'}{h_2} = \reduce{s'}{(h_1 \uplus h_2)}$. Thus, $(s', h') \models \phi$.

        \item[$(\Leftarrow)$]
        By the assumption, there exist disjoint heaps $h_1$, $h_2$ such that $\reduce{s'}{h} = h_1 \uplus h_2$ and $(s',h_i) \models \psi_i$ for $i=1,2$. Let $\H_i = \chunks(s',h_i)$ for $i=1,2$ and let $\H = \chunks(s',h)$. Since $\reduce{s'}{h} = h_1 \uplus h_2$, for each chunk $c \in \H_1 \cup \H_2$ there exists a unique chunk $c' \in \H$ such that $\reduce{s'}{c'} = c$. Let $h'_1$ and $h'_2$ be defined as $h'_i = \biguplus \{c' \;|\; c \in \chunks(s,h_i)\}$. In other words, we compose $h'_i$ from chunks of $h$ that were reduced to chunks of $h_i$. By construction, it holds that $\reduce{s'}{h'_i} = h_i$ and therefore by IH $(s',h'_i) \models \psi_i$. It also holds that $h'_1 \uplus h'_2 = h \neq \bot$. Thus $(s',h) \models \psi_1 \star \psi_2$. \qed
        \end{itemize}
\end{itemize}
\end{proof}

We split the proof of Theorem \ref{theorem:small-models}, into several lemmas giving bounds on components of canonical models.



\begin{lemma}[Bound on the number of allocated locations]
\label{lemma:domain-bound}
Let $\phi$ be a formula and let $(s,h)$ be its canonical model. Then, $|\dom(h)| \leq \floor{\sum_{x \in \vars(\phi)} \chunksize{x}}$.
\end{lemma}

\begin{proof}
From Lemma \ref{lemma:chunk-decomposition}, we have that $h$ is fully decomposed into its chunks. Since each chunk of a sort $S$ must contain an allocated root variable, the number of $S$-chunks is bounded by the number of variables of the sort $S$. We will now discuss the upper bounds on sizes of reduced chunks. A pointer-chunk always has the size~$1$. The size of a chunk of a sort $S \in \{\lsSort, \nlsSort\}$ is $2$ by the definition of the reduction. The size of a reduced chunk of a sort~$\dlsSort$ is $3$, but such a chunk needs to contain two allocated variables ($\dlspred$-chunk allocates both $x$ and $x'$). As the result, each variable of the sort $\dlsSort$ adds 1.5 to the total bound. 

In the worst case, almost are chunks are proper predicate chunks. The only exception is the case when the number of DLL variables is odd. In this case, the last ``unpaired" variable can create just a pointer-chunk, not a proper $\dls$-chunk. This is the reason why the sum is rounded down. \qed
\end{proof}

In the next step, we show that positive models do not have any unlabelled dangling location (a location $\ell$ is dangling if $\ell \in \img(h)$, but $\ell \not\in \dom(h)$).

\begin{lemma}
\label{lemma:dangling-labelled}
For a positive model $(s,h)$, it holds that $\img(h) \setminus \dom(h) \subseteq \img(s)$.
\end{lemma}

\begin{proof}
Since $(s,h)$ is positive, there is a formula $\phi$ such that $(s,h) \models \phi$. The claim is proved by structural induction on $\phi$. The base cases are trivial as the dangling locations of $x \mapsto \langle \f_i : y_i\rangle_{i \in I}$, $\lspred$, $\dlspred$ and $\nlspred$ are precisely $\{y_i\}_{i \in I}$, $\{y\}$, $\{y, y'\}$ and $\{y, z\}$, respectively. The induction step follows directly from IH in all the cases, for the case of separating conjunction this is because a composition of two heaps cannot introduce a new dangling location. \qed
\end{proof}

\begin{proof}[Theorem \ref{theorem:small-models}] Let $\phi$ be a satisfiable formula. From Theorem~\ref{theorem:reduction}, we know that there exists a canonical model $(s,h) \models \phi$. Recall that the set of locations in $(s,h)$ is defined as $\locations(s,h) = \dom(h) \cup \img(h) \cup \img(s)$, which can be simplified to $\locations(s,h) = \dom(h) \cup \img(s)$ using Lemma \ref{lemma:dangling-labelled}.

From Lemma~\ref{lemma:domain-bound}, we already have that $|\dom(h)| \leq \floor{\sum_{x \in \vars(\phi)} \chunksize{x}}$. When computing this bound, we have already assigned at least one location to every variable and thus no additional locations are needed for $\img(s)$, except one for $\nil$ which may not appear in $\phi$, but we require $s(\nil)$ to be present in every stack-heap model. This gives us the bound $|\locations(s, h)| \leq 1 + \floor{\sum_{x \in \vars(\phi)} \chunksize{x}}$. \qed
\end{proof}

    \newpage
\section{Translation Correctness}

This section presents proofs omitted in Section \ref{section:translation}. First, we develop the concept of correspondence of stack-heap and first-order models and its properties in Section \ref{appendix:model-correspondence} and then we present correctness proofs in Section \ref{appendix:translation-correctness}. For this section, we fix a separation logic formula~$\phi$.

\subsection{Model Correspondence}
\label{appendix:model-correspondence}
We define the translation's signature \mbox{$\Sigma_\phi =\{(x)_{x \in \vars(\phi)}, (h_\f)_{\f \in \Fields}, (D_S)_{S \in \Sort}, D\}$}. We call a first-order model over $\Sigma_\phi$ a \textit{model of SMT encoding} (SMT model for short) if $\M \models \mathcal{A}_\phi$ and \textit{canonical} if for all $\ell \in \mathsf{L} \setminus D^\M$ and all $\f \in \Fields$, it holds that $h_\f[\ell]^\M = \mathsf{loc}^\nil$. We call models $(s,h)$ and $\M$ \textit{corresponding} if $(s,h) = \itranslate{\phi}{\M}$. 

\begin{corollary}[Unique corresponding models]
For a stack-heap model $(s, h)$, there is the unique canonical corresponding SMT model $\M$, and vice versa.
\end{corollary}

\begin{proof}
For a first-order model $\M$, the stack-heap model $\itranslate{\phi}{\M}$ is the unique corresponding model. For a stack-heap model $(s,h)$, the unique corresponding model is the canonical model $\M$ with $\itranslate{\phi}{\M}$. \qed
\end{proof}

We now define an operation of composition of SMT models which mimics the composition of stack-heap models with the same stack and disjoint heaps. 

\begin{definition}[Compatible models]
Let $F \subseteq \Locs$ be a set of locations. SMT models $\M_1$ and $\M_2$ are \mbox{$F$-\emph{compatible}} if the following conditions hold:
\begin{itemize}
    \item $D^{\M_1} \cap D^{\M_2} \subseteq F$,
    \item For all $x \in \vars(\phi)$, $x^{\M_1} = x^{\M_2}$,
    \item For all $\ell \in F$ and all $\f \in \Fields$, $h_\f[\ell]^{\M_1} = h_\f[\ell]^{\M_2}$.
\end{itemize}
\end{definition}

\noindent We say that models are compatible if they are $\emptyset$-compatible. Observe that compatible models correspond to stack-heap models with the same stack and disjoint heaps.

\begin{definition}[Model composition]
The composition of $\M_1$ and $\M_2$, $\M_1 \oplus \M_2$, is defined iff $\M_1$ and $\M_2$ are compatible as $\M_1$ except:
\begin{itemize}
    \item $D^{\M_1 \oplus \M_2} = D^{\M_1} \cup D^{\M_2}$
    \item $h_\f[\ell]^{\M_1 \oplus \M_2} = 
        \begin{cases}
            h_\f[\ell]^{\mathcal{M}_1} \quad &\text{ if $\ell \in D^{\mathcal{M}_1}$,}\\
            h_\f[\ell]^{\mathcal{M}_2} \quad &\text{ if $\ell \in D^{\mathcal{M}_2}$,}\\
            \mathsf{loc}^{\nil} &\text{ otherwise.}
        \end{cases}$
\end{itemize}
\end{definition}

\begin{corollary}
\label{corollary:model-composition}
Let $(s,h_i)$ and $\M_i$ be corresponding models. Then \mbox{$(s,h_1 \uplus h_2)$} corresponds to $\M_1 \oplus \M_2$.
\end{corollary}

\begin{proof}
First, observe that $\M_1 \oplus \M_2$ is SMT model iff $\M_1$ and $\M_2$ are SMT models. The rest of the claim follows directly from the definitions of model correspondence and model composition. \qed
\end{proof}

The translation $\translate{\psi}{F}$ has the invariant that it cannot distinguish $F$-compatible models.

\begin{lemma}
\label{lemma:translation-invariant} For $F$-compatible SMT models $\M_1$ and $\M_2$, it holds that:
$$\M_1 \models \translate{\psi}{F} \iff \M_2 \models \translate{\psi}{F}.$$
\end{lemma}

\begin{proof}
By the definition, $F-$compatible models $\M_1$ and $\M_2$ can differ only in their interpretation of the symbol $D$ and of array images $h_\f[\ell]$ such that $\ell \not \in F$. It can be easily checked that $D$ does not occur in $\translate{\psi}{F}$ which therefore does not restrict its interpretation at all. It remains to show that $\translate{\psi}{F}$ also does not restrict the interpretation of $h_\f[\ell]$ such that $\ell \not \in F$ which can be done by a straightforward structural induction on~$\phi$. \qed
\end{proof}

The following lemma captures the correctness of translation of the separating conjunction.

\begin{lemma}
\label{lemma:star-lemma}
For compatible SMT models $\M_1$ and $\M_2$, the following are equivalent:
\begin{itemize}
    \item $\M_1 \models \translate{\psi_1}{F_1} \land D = F_1$ \text{ and } $\M_2 \models \translate{\psi_2}{F_2} \land D = F_2$,
    \item $\M_1 \oplus \M_2 \models \translate{\psi_1}{F_1} \land \translate{\psi_2}{F_2} \land D = F_1 \cup F_2$.
\end{itemize}
\end{lemma}

\begin{proof} 
Since $\M_1$ and $\M_2$ are compatible, we know that $\M_1 \oplus \M_2$ is defined. Further, it holds that $\M_1 \oplus \M_2$ is $F_i$-compatible with $\M_i$. The rest of the claim follows from the fact that $\M_1 \oplus \M_2 \models \translate{\psi_i}{F_i}$ iff $\M_i \models \translate{\psi_i}{F_i}$ by Lemma \ref{lemma:translation-invariant}. \qed

\end{proof}

\subsection{Proof of Translation Correctness}
\label{appendix:translation-correctness}

\begin{proof}[Lemma \ref{lemma:star_simplified}]$\;$
\begin{itemize}[leftmargin=2.5em,labelindent=1em]
    \item[$(\Rightarrow)$] 
    Assume that $(s,h) \models \psi_1 \star \psi_2$, then there exist disjoint heaps $h_1$ and $h_2$ with $h = h_1 \uplus h_2$ and $(s, h_i) \models \psi_i$. Then $\dom(h_i) \in \mathcal{F}_i$ for $i = 1, 2$ because it is a footprint of $\psi_i$ in $(s,h)$. The right-hand side then holds for $F_i \coloneqq \dom(h_i)$.
    \item[$(\Leftarrow)$]
    Assume that $F_1 \in \mathcal{F}_1$, $F_2 \in \mathcal{F}_2$ are disjoint sets satisfying $(s, h|_{F_i}) \models \psi_i$ and $F_1 \cup F_2 = \dom(h)$, then $(s, h) \models \psi_1 \star \psi_2$. \qed
\end{itemize}
\end{proof}

We will now define the semantics of inductive predicates in terms of paths in stack-heap graphs. For this we will use a property that such paths are uniquely determine.

\begin{corollary}[Path uniqueness]
\label{corollary:path-uniqueness}
Let $x$ and $y$ be locations in a model $(s,h)$. If there is a path $\gpath{\sigma}{x}{y}{\f}$, then it is uniquely determined as $\sigma = \langle x, h(x, \f), \ldots, h^{|\sigma|}(x, \f) \rangle$.
\end{corollary}

\begin{proof}
The claim is a direct consequence of the fact that $\f$-successor of a location~$\ell$ in $G[(s,h)]$ is uniquely determined as $h(\ell, \f)$; and of paths being defined as acyclic. \qed
\end{proof}

\noindent For reasoning about inductive predicates, we formally define the sorts of locations as follows:
\vspace{-4mm}
\begin{align*}
\Locs_\ls    &= \{ \ell \in \Locs \;|\; h(\ell, \n) \neq \bot, h(\ell, \p) = h(\ell, \t) = \bot \}\\
\Locs_\dls   &= \{ \ell \in \Locs \;|\; h(\ell, \n) \neq \bot, h(\ell, \p) \neq \bot, h(\ell, \t) = \bot \} \\
\Locs_\nls   &= \{ \ell \in \Locs \;|\; h(\ell, \n) \neq \bot, h(\ell, \t) \neq \bot, h(\ell, \p) = \bot \}
\end{align*}
\noindent For a path $\sigma$, a field $\f$ and a location $z$, we further define the \emph{nested domain} of $\sigma$ in $G[(s,h)]$ as $\nesteddom(\sigma, \f, z) = \dom(\sigma) \cup \bigcup \{\sigma' \;|\; \exists \ell \in \dom(\sigma).\; \gpath{\sigma'}{\ell}{s(z)}{\f}\}$.
\begin{lemma}[Path semantics of inductive predicates]
\enlargethispage{6mm}
\label{lemma:path-semantics}
Let $(s,h)$ be a model and let $\pi$ be an inductive predicate. The following claims hold:
\begin{itemize}
    \item $(s,h) \models \lspred$ iff there is a path $\gpath{\sigma}{x}{y}{\n}$ s.t. $\dom(h) = \dom(\sigma) \subseteq \Locs_\ls$.\\
    \item $(s,h) \models \dlspred$ iff either $s(x) = s(y)$, $s(x') = s(y')$ and $\dom(h) = \emptyset$, or $s(x) \neq s(y)$, $s(x') \neq s(y')$ and there exists a path $\gpath{\sigma}{x}{y}{\n}$ such that:
    \begin{enumerate}
        \item $\dom(h) = \dom(\sigma) \subseteq \Locs_\dls$,
        \item $\ell \in \dom(h)$ and $\ell \neq s(x')$ implies that $h(\p, h(\n, \ell)) = \ell$,
        \item $h(s(x'), \n) = s(y)$ and $s(x') \in \dom(h)$,
        \item $h(s(x), \p) = s(y')$ and $s(y') \not \in \dom(h)$.\\
    \end{enumerate}
    \item $(s,h) \models \nlspred$ iff there exists a path $\gpath{\sigma}{x}{y}{\t}$ such that:
    \begin{enumerate}
        \item $\ell \in \dom(h)$ implies that there exists a path $ \gpath{\sigma'}{\ell}{z}{\n}$,
        \item For all distinct $\ell_1$ and $\ell_2$, if $\{\ell_1, \ell_2\} \subseteq \nesteddom(\sigma, \n, z)$ and $h(\ell_1, \n) = h(\ell_2, \n)$, then $h(\ell_1, \n) \not \in \nesteddom(\sigma, \n, z)$.
        \item $\dom(h) = \nesteddom(\sigma, \n, z)$,
        \item $\dom(\sigma) \subseteq \Locs_\nls$ and $\dom(h) \setminus \dom(\sigma) \subseteq \Locs_\ls$.
    \end{enumerate}
\end{itemize}
\end{lemma}

\begin{proof}
We only prove the more complex implication ($\Leftarrow$). For the other direction, it can be intuitively checked that models of inductive predicates satisfy the given conditions and its proof is analogical. We proceed by induction on $|\sigma|$ in all the cases.

\paragraph{SLL.} If $|\sigma| = 0$, then the claim holds as $s(x)~=~s(y)$ and $\dom(h) = \dom(\sigma) = \emptyset$.  Assume that $|\sigma| = n + 1$. Since there exists a path \mbox{$\gpath{\sigma}{s(x)}{s(y)}{\n}$} of length $n + 1$ with $s(x) \in \dom(\sigma) \subseteq \Locs_\ls$, there also exists a path $\gpath{\sigma'}{h(s(x), \n)}{s(y)}{\n}$. Then, by IH, $(s,h) \models \exists u.\; x^\lsSort \pto u \star \ls(u, y)$. Moreover, $s(x) \neq s(y)$ because $\sigma$ is non-empty. Thus $(s,h) \models \lspred$.

\paragraph{DLL.} If $|\sigma| = 0$, then the claim trivially holds. 
If $|\sigma| = 1$, then $h(s(x), \n) = s(y)$ and $\dom(\sigma) = \{s(x)\} \subseteq \Locs_\dls$. Hence, by the conditions (3) and (4), we have that $(s,h) \models x^\dlsSort \pto \cdls{y}{y'} \star x \neq y \star x \neq y'$ where the first disequality follows from non-emptiness of the path $\sigma$ and the second follows from the fact that $s(x') \in \dom(\sigma)$ while $s(y') \not \in \dom(\sigma)$ by (4). Thus, $(s,h) \models \dlspred$.
    
If $|\sigma| = n + 2$, then $\{s(x), h(s(x), \n)\} \subseteq \dom(\sigma) \subseteq \Locs_\dls$ and there exists a path $\gpath{\sigma'}{h(s(x), \n)}{s(y)}{\n}$. We will show that for $\sigma'$, the conditions (1) -- (4) hold. This is trivial for (1), (2) and (4) because $\dom(\sigma') = \dom(\sigma) \setminus \{s(x)\}$. The condition (3) follows from the fact that $s(x) \neq s(x')$ because $h(s(x), \n) \neq h(s(x'), \n)$ (the first location is in $\dom(\sigma)$ by assumption while the second is not as it is equal to $s(y)$ by~(3)).  Hence, by IH, $(s,h) \models \exists u.\; x^\dlsSort \pto \cdls{u}{y'} \star \dls(u, y, x', x) \star x \neq y \star x' \neq y'$. Thus $(s,h) \models \dlspred$.

\paragraph{NLL.} The base case is analogical to the case of SLL. If $|\sigma| = n + 1$, then $s(x) \in \dom(\sigma)$ and there is a path $\gpath{\sigma_\t}{h(s(x), \t)}{ s(y)}{\t}$ with $|\sigma_\t| < n$. As $\dom(\sigma_\t) \subseteq \dom(\sigma)$ we have that for a sub-heap $h_\t$ with $\dom(h_\t) = \nesteddom(\sigma', \n, z)$, $(s,h_\t) \models \exists t. \nls(t, y, z)$ by IH. By (1), we further have that there is a path $\gpath{\sigma_\n}{h(s(x), \n)}{s(z)}{\n}$ for which $\dom(\sigma_n) \subseteq \Locs_\ls$ by~(4). Hence $(s,h_\n) \models \exists n.\; \ls(n, z)$. By (2), we have that $h_\t$ and $h_n$ are disjoint. Combining together all the claims derived so far, we obtain that $(s,h) \models \exists t, n.\; x^\nlsSort \pto \cnls{n}{t} \star \ls(n,z) \star \nls(t, y, z) \star x \neq y$ where the last disequality follows from non-emptiness of the path $\sigma$. Thus, $(s,h) \models \nlspred$.
\end{proof}

\begin{lemma}[Translation of path semantics]
\enlargethispage{8mm}
\label{lemma:auxiliary-predicates}
Let $(s,h)$ and $\M$ be corresponding models. The following claim holds:
\begin{enumerate}
    \item \!\!$\mathcal{M} \models \mathsf{reach}^{[m,n]}(h_\f,x,y)$ iff 
    there is a path $\gpath{\sigma}{s(x)}{s(y)}{\f}$ with $m \leq |\sigma| \leq n$.
\end{enumerate}
\noindent Further, let $\psi$ be a formula with one free variable $v$. Assume that there exists a path $\gpath{\sigma}{s(x)}{s(y)}{\f}$ with $m \leq |\sigma| \leq n$, then the following claims hold:
\begin{enumerate}
    \item[2.] $\path^{[m,n]}_S(h_\f, x, y)^\M = \dom(\sigma)$,
    \item[3.] $\M \models \allpathBounds{h_\f}{x}{v} v \in \path^{[m,n]}_S(h_\f, x, y) \implies \psi$ iff $\M \models \psi[\ell/v]$ for all $\ell \in \dom(\sigma)$.
\end{enumerate}
\noindent Further, assume that for each $\ell \in \dom(\sigma)$ there is $\gpath{\sigma'}{\ell}{s(z)}{\f'}$ s. t. $|\sigma'| \leq k$, then:
\begin{enumerate}
    \item[4.] $\path^{[m,n]}_N(h_\f, h_{\f'}, x, y, z)^\M = \nesteddom(\sigma, \f', z)$,
    \item[5.] $\M \models \allpathBounds{h_\f}{x}{v'} \mathbb{P}^{\leq k}_{(h_\f', v')}\;v.\; v \in \path^{[m,n]}_N(h_\f, h_{\f'}, x, y, z) \implies \psi$ iff \mbox{$\M \models \psi[\ell/v]$}
    for all $\ell \in \nesteddom(\sigma, f', z)$.
\end{enumerate}
\end{lemma}

\begin{proof} The claims (1), (2) and (4) directly follow from Corollary \ref{corollary:path-uniqueness}. The claim (3) follows from the fact that the size of the path $\dom(\sigma)$ is at most $n$. The proof of (5) is analogical to (3) using the fact that size of each inner path is at most $k$.
\end{proof}

\begin{proof}[Lemma \ref{lemma:footprints}]
We write $\ffootprints{\phi}$ instead of $\{F^\M \;|\; F \in \mathsf{FP}^{\#}(\phi)\}$. We will show that $\footprints{\phi} \subseteq \ffootprints{\phi}$ writing just $\Footprints{\phi} \subseteq \FFootprints{\phi}$ for short.

If $(s,h) \not\models \phi \star \true$, then there does not exists $F \subseteq \dom(h)$ such that $(s,h|_F) \models \phi$ and the claim trivially holds since $\Footprints{\phi} = \emptyset$. Assuming that $(s,h) \models \phi \star \true$, we proceed by structural induction on $\phi$.

\begin{itemize}
    \item \textit{Base cases.} For all the base cases, it holds that $\Footprints{\phi} = \FFootprints{\phi} = \{F\}$ where~$F$ is defined as follows. If $\phi$ is a pure atom or a pointer assertion $x \mapsto \wildcard$, then $F = \{\emptyset\}$ or $F = \{s(x)\}$, respectively. If $\phi$ is an inductive predicate, its unique footprint is given by Lemma~\ref{lemma:path-semantics} and its encoding in SMT is correct by Lemma~\ref{lemma:auxiliary-predicates}.

    \item \textit{Induction steps.} Assume that $\phi$ is a binary connective with operands $\psi_1$ and $\psi_2$. By induction hypothesis, $\Footprints{\psi_i} \subseteq \FFootprints{\psi_i}$ for $i = 1,2$.

    \vspace{2pt}

    \begin{itemize}
        \item If $\phi \triangleq \psi_1 \land \psi_2$, let $i$ denote the operand for which the set $\FFootprints{\psi_i}$ has smaller cardinality. Then:
        $$\Footprints{\phi} \subseteq \Footprints{\psi_1} \cap \Footprints{\psi_2} \subseteq \Footprints{\psi_i} \subseteq \FFootprints{\psi_i} = \FFootprints{\phi},$$

        \noindent where the first inclusion follows from the definition of footprints and the last follows from IH.

        \item If $\phi \triangleq \psi_1 \gneg \psi_2$, then the proof is analogical to the previous case with $i= 1$.

        \item If $\phi \triangleq \psi_1 \lor \psi_2$, then we have:
        $$\Footprints{\phi}
                \subseteq \Footprints{\psi_1} \cup \Footprints{\psi_2}
                \subseteq \FFootprints{\psi_1} \cup \FFootprints{\psi_2}
                = \FFootprints{\phi},
        $$

        \noindent where the first inclusion follows from the definition of footprints and the second from IH.
    
        \item If $\phi \triangleq \psi_1 \star \psi_2$, then we have:
        \begin{align*}
        \Footprints{\phi} 
            &\subseteq \{F_1 \cup F_2 \;|\; F_1 \cap F_2 = \emptyset \text{ and } F_1 \in               \Footprints{\psi_1}, F_2 \in \Footprints{\psi_2}\}\\
            &\subseteq \{F_1 \cup F_2 \;|\;  F_1 \in \Footprints{\psi_1}, F_2 \in                       \Footprints{\psi_2}\}\\
            &\subseteq \{F_1 \cup F_2 \;|\;  F_1 \in \FFootprints{\psi_1}, F_2 \in                       \FFootprints{\psi_2}\}\\
            &= \FFootprints{\phi}.
        \end{align*}

    \end{itemize}
\end{itemize}
\end{proof}

\newpage
\enlargethispage{6mm}
\begin{proof}[Theorem \ref{theorem:correctness}]
We reformulate the original claim to obtain a stronger inductive hypothesis. We prove that for arbitrary SL formula $\psi$ and corresponding models $(s,h)$ and $\M$, it holds that:
$$\M \models \translate{\psi}{F} \land D = F \iff (s,h) \models \psi.$$

\noindent We proceed by structural induction on $\phi$:
\begin{itemize}
\item For a pure formula $\psi \triangleq x \bowtie y$ where $\bowtie \;\in \{=, \neq\}$, we have:
\begin{align*}
(s,h) \models \psi
    & \iff s(x) \bowtie s(y) \land \dom(h) = \emptyset && \text{(SL semantics)}\\
    & \iff \M \models x \bowtie y \land D = \emptyset  && \text{(Model correspondence)}\\
    & \iff \M \models \translate{\psi}{F} \land D = F  && \text{(Formula translation)}
\end{align*}

\item For a pointer assertion $\psi \triangleq x^S \pto \langle \f_i : f_i \rangle_{i \in I}$, we have that $(s,h) \models \psi$
\begin{align*}
    & \iff \dom(h) = \{s(x)\} \land h(s(x)) = \langle \f_i : f_i \rangle_{i \in I} 
        && \text{(SL semantics)}\\
    & \iff \M \models D = \{x\} \land x \in D_S \land \bigwedge\!_{i \in I} \; h_{\f_i}[x] = f_i
        && \text{(Model correspondence)}\\
    & \iff \M \models \translate{\psi}{F} \land D = F                                                   
        && \text{(Formula translation)}
\end{align*}

Notice that the condition $x \in D_S$ is needed to correctly reconstruct $h(x)$ and in our translation, it is ensured by axioms $\mathcal{A}_\phi$.

\vspace{4pt}

\item The cases of inductive predicates follow from Lemma \ref{lemma:path-semantics} and Lemma \ref{lemma:auxiliary-predicates}.

\item For a boolean connective $\psi \triangleq \psi_1 \bowtie \psi_2$ where $\bowtie \;\in \{\land, \lor, \gneg\}$, we have:
\begin{align*}
(s,h) \models \psi 
    & \iff (s, h) \models \psi_1 \bowtie (s,h) \models \psi_2 && \text{(SL semantics)}\\
    & \iff \M \models \translate{\psi_1}{F} \bowtie \translate{\psi_2}{F} \land D = F && \text{(Induction hypothesis)}\\
    & \iff \M \models \translate{\psi}{F} \land D = F  && \text{(Formula translation)}
\end{align*}

\item For a separating conjunction $\psi_1 \star \psi_2$, we proceed as follows. Let $\mathcal{F}_i = \ffootprints{\psi_i}$. From Lemma \ref{lemma:star_simplified} and Lemma \ref{lemma:footprints}, we have that $(s,h) \models \phi$ iff
    $$\bigvee_{F_1 \in \mathcal{F}_1} \;
        \bigvee_{F_2 \in \mathcal{F}_2} \;
        \bigwedge_{i = 1, 2} (s, h|_{F_i}) \models \psi_i \;\land\; F_1 \cap F_2 = \emptyset \;\land\; F_1 \cup F_2 = \dom(h).$$

    Let $\M_i$ be the corresponding model to $(s,h|_{F_i})$ for $i = 1, 2$. By IH:
    $$\bigvee_{F_1 \in \mathcal{F}_1} \;
        \bigvee_{F_2 \in \mathcal{F}_2} \;
        \bigwedge_{i = 1, 2} (\M_i\models \translate{\psi_i}{F_i} \land D = F_i)  \;\land\; F_1 \cap F_2 = \emptyset \;\land\; F_1 \cup F_2 = D^\M.$$
    Using Corollary \ref{corollary:model-composition}, we have that $\M = \M_1 \oplus \M_2$. Applying Lemma \ref{lemma:star-lemma}, we have an equivalent claim:
    $$\bigvee_{F_1 \in \mathcal{F}_1} \;
        \bigvee_{F_2 \in \mathcal{F}_2} \;
        \M \models \translate{\psi_1}{F_1} \land \translate{\psi_2}{F_2} \;\land\; F_1 \cap F_2 = \emptyset \;\land\; F_1 \cup F_2 = D^\M.$$
    Finally, this is equivalent to the following claim:
    $$\M \models \bigvee_{F_1 \in \mathcal{F}_1} \;
        \bigvee_{F_2 \in \mathcal{F}_2} \;
        \bigwedge_{i = 1, 2} \translate{\psi_i}{F_i} \;\land\; F_1 \cap F_2 = \emptyset \;\land\; F_1 \cup F_2 = D,$$
    
    which is the case iff $\M \models \translate{\psi}{F} \land D = F$.
    \qed
\end{itemize}
\end{proof}

    \newpage
\section{Complexity of BSL Satisfiability}
\label{appendix:complexity}

\enlargethispage{6mm}

In this section, we give details for the complexity results from Section \ref{section:complexity}. We use $n = \bound(\phi)$ to denote the location bound of the formula $\phi$ and note that $n$ is linear w.r.t. the size of  $\phi$.
\begin{proof}[Theorem \ref{theorem:complexity}]
We prove that $\mathtt{SatQuantif}$ produces formulae of polynomial size, in particular $\mathcal{O}(n^3)$ for SLLs and DLLs, and $\mathcal{O}(n^5)$ for NLLs. We proceed by analysis of asymptotic sizes of auxiliary predicates used for the translation of inductive predicates. Let $x$ be a term of the size $\mathcal{O}(n)$ and let $h, h', y, z$ be terms of the size $\mathcal{O}(1)$. This assumption is justified by the fact that parameters of all predicates in the translation except the roots of considered paths (denoted by $x$) are variables. The root $x$ can be a term of the form $h^k_\n[v]$ or $h^k_\n[h^\ell_\t[v]]$ (where $v$ is a variable) and is therefore always of linear size because the numbers $k$ and $\ell$ are never greater than the location bound $n$. The size of auxiliary predicates is then:
\begin{itemize}
    \item $|\reach^{=k}(h, x, y)| \in \mathcal{O}(n)$ when $k \leq n$ because its size is dominated by $x$.
    \item $|\reach^{[0, n]}(h, x, y)| \in \mathcal{O}(n^2)$ as it contains $n$-times $\reach^{=k}(h, x, y)$ for \mbox{$k \leq n$}.
    \item $|\path^{=k}_S(h, x, y)| \in \mathcal{O}(n^2)$ when $k \leq n$ because it is a union of $k$ terms of the form $\{h^i[x]\}$ for $i \leq k \leq n$ and thus of the size $\mathcal{O}(n)$.
    \item $|\path^{[0, n]}_S(h, x, y)| \in \mathcal{O}(n^3)$ because it consists of $\mathcal{O}(n)$ branches, each of them dominated by the size of $\path^{=k}_S(h, x, y)$ for $k \leq n$ which is $\mathcal{O}(n^2)$.
    \item $|\path^{=k}_N(h, h', x, y, z)| \in \mathcal{O}(n^4)$ when $k \leq n$ because it is a union of $k$ terms of the form $\path^{[0, n]}_S(h', h^i[x], y)$ for $i \leq k \leq n$ and thus of the size $\mathcal{O}(n^3)$.
    \item $|\path^{[0, n]}_N(h, h', x, y, z)| \in \mathcal{O}(n^5)$ because it consists of $\mathcal{O}(n)$ branches, each of them dominated by the size of $\path^{=k}_N(h, h', x, y, z)$ for $k \leq n$ which is $\mathcal{O}(n^4)$.
    
\end{itemize}
Translation of SLLs and DLLs is dominated by the size of the $\path_S$ term and translation of NLLs is dominated by the size the $\path_N$ term. Note, that this holds for the translation version using universal quantifiers over locations. When the optimised version with path quantifiers is used for NLLs, we may obtain asymptotically larger formulae, but still polynomial.

For all the other atomic formulae, the size of translated formula is constant. The translation of all binary connective is linear in the size of their operands (also in the case of separating conjunction when the method $\mathtt{SatQuantif}$ is used). The size of axioms is also linear. Thus, the final size of the translated formula is dominated by the translation of inductive predicates. \qed

\end{proof}

\begin{proof}[Theorem \ref{theorem:complexity}]
We give the full definition of the reduction from QBF based \mbox{on~\cite{CC_results}}. Let $F \triangleq Q_1 x_1 \ldots Q_m x_m.\; F'$ be a ground formula with variables $X = \{x_1, \ldots, x_m\}$ (w.l.o.g., they are all distinct), quantifiers $Q_i \in \{\forall, \exists\}$, and propositional body $F'$. Recall that for a set of variables~$Y$, we define $\mathsf{arbitrary}[Y] \triangleq \bigstar_{y \in Y} (y \pto \nil \lor \emp)$. The reduction of~$F$ is defined as $\red{F} \hspace{1pt}\land\hspace{1pt} \bigstar_{x \in X} x \mapsto \nil$ where $\red{\cdot}$ is defined as:
\begin{align*}
\red{x}                   &\triangleq \mathsf{arbitrary}[X] \star x \mapsto \nil
    &&\hspace{-40pt}\red{F \land G}      \triangleq \red{F} \land \red{G}\\
\red{\neg x}              &\triangleq \mathsf{arbitrary}[X \setminus \{x\}]
    &&\hspace{-40pt}\red{F \lor G}       \triangleq \red{F} \lor \red{G}\\
\red{\exists x.\; F}      &\triangleq (x \mapsto \nil \lor \emp) \star \red{F}
    &&\hspace{-40pt}\red{\neg F}         \triangleq \mathsf{arbitrary}[X] \gneg \red{F}\\
\red{\forall x.\; F}      &\triangleq \mathsf{arbitrary}[X] \gneg ((x \mapsto \nil \lor \emp) \star \red{\neg F})
\end{align*} 
\end{proof}



    \newpage
\section{Bitvector Encoding}
\label{appendix:bitvectors}

In this section, we sketch the bitvector encoding that we have used for the experimental evaluation in Section \ref{section:experiments}. We assume a subset of the theory of \textit{fixed sized bitvectors}\footnote{\url{https://smtlib.cs.uiowa.edu/theories-FixedSizeBitVectors.shtml}} with the signature containing the following symbols: $\bvor$ (bit-wise or), $\bvand$ (bit-wise and), $\overline{\;\cdot\;}$~(bit-wise negation), $\shiftLeft$ (left shift), and $b[i]$ (extraction of the $i$-th least significant bit of the bitvector~$b$, interpreted as a boolean value). We further write $k_w$ to denote the unsigned representation of integer~$k$ as a bitvector of width $w$.

\vspace*{-1mm}\paragraph{Set encoding.}

The main idea of the bitvector encoding is to represent location sets as bitvectors where the boolean value of $i$-th least significant bit represents the membership of $i$-th location in the set. We therefore define the width of considered bitvectors $w = n$, where $n$ is the location bound of the input formula. To reduce the state space, one could consider computing a lower width by excluding locations that cannot be allocated (such as $\nil$), but we have not done this so far. We encode the set operations used in our encoding as follows.
\begin{align*}
\begin{split}
     \emptyset &\triangleq 0_w\\
     X_1 \cup X_2 &\triangleq X_1 \bvor X_2\\
     X_1 \cap X_2 &\triangleq X_1 \bvand X_2
\end{split}
\begin{split}
     \{x_1, \ldots, x_m\} &\triangleq 1_w \shiftLeft x_1 \bvor \cdots \bvor 1_w \shiftLeft x_m\\
     x \in X &\triangleq X[x] \\
     X_1 \subseteq X_2 &\triangleq \overline{X_1} \bvor X_2 = \overline{0_w}
\end{split}
\end{align*}

\vspace*{-1mm}\paragraph{Memory model axioms.}

We represent locations also as bitvectors of the same width~$w$. For consistency, we need to ensure that each location used in a model is interpreted as a bitvector lesser than~$n$ (greater bitvectors would not fit into sets represented as bitvectors which have size $w = n$). We first define $\L = \{0_w, 1_w, \ldots, (n-1)_w\}$ and axioms defining sets representing location sorts from Section \ref{section:encoding-memory-model} as:
\begin{align*}
    \mathcal{A}_{D_\lsSort} \triangleq
        D_\lsSort  &= \{0_w\} \cup \{1_w \ldots, (n_\lsSort)_w\},\\
    \mathcal{A}_{D_\dlsSort} \triangleq
        D_\dlsSort &= \{0_w\} \cup \{(n_\lsSort + 1)_w, \ldots, (n_\lsSort + n_\dlsSort)_w\},\\
    \mathcal{A}_{D_\nlsSort} \triangleq
        D_\nlsSort &= \{0_w\} \cup \{(n_\lsSort + n_\dlsSort + 1)_w, \ldots, (n_\lsSort + n_\dlsSort + n_\nlsSort)_w\}.
\end{align*}
\noindent Notice that those axioms will ensure that all locations referenced by some variable are interpreted as bitvectors lesser that $n$ (since $n = n_\lsSort + n_\dlsSort + n_\nlsSort + 1)$. To ensure the same for anonymous locations, we add the following axiom to restrict all heap images to be lesser than $n$:
\begin{align*}
\mathcal{A}_\mathrm{heap} \triangleq
    \bigwedge_{\f \in \Fields} \;
    \bigwedge_{\ell \in \L} \;
    h_\f[\ell] < n_w.
\end{align*}

\noindent Putting it all together, the axioms in the bitvector encoding are defined as:
$$\mathcal{A}_\phi \triangleq 
    \nil = 0_w \;\land\;
    \nil \not \in D \land\; \mathcal{A}_\mathrm{heap} \;\land \bigwedge_{S \in \Sort}\!\!\! \bigl( \mathcal{A}_{D_S} \land
\!\!\!\!\!\bigwedge_{x \in \vars_S(\phi)} \!\!\!\!\!x \in D_S\bigl).$$

\fi

\end{document}